\documentclass{vldb}
\usepackage{soul}
\usepackage{graphicx}
\usepackage{balance}  
\usepackage{enumitem}
\usepackage{subfigure}
\usepackage{caption}
\usepackage{amsmath, bm}
\usepackage{amsfonts}
\usepackage{multirow}
\usepackage{multicol}

\usepackage{algorithm} 
\usepackage[algo2e,ruled,noend,linesnumbered]{algorithm2e}
\DontPrintSemicolon
\usepackage{color}
\usepackage{caption}
\usepackage{pifont}
\usepackage{hyperref}

\usepackage{comment}

\newtheorem{definition}{definition}

\newtheorem{prerequisite}{Prerequisite}
\newtheorem{theorem}{theorem}

\begin{document}


\title{Transparent Concurrency Control: Decoupling Concurrency Control from DBMS}



%
%
%
%

\numberofauthors{1} 

\author{
%
%
\alignauthor
\vspace{-1cm}
Ningnan Zhou$^\dag$~~~~Xuan Zhou$^\ddag$~~~~Kian-lee Tan$^\S$~~~~Shan Wang$^\dag$ \\
       \affaddr{$^\dag$ \textit{DEKE Lab, Renmin University of China, Beijing, China}}\\
       \affaddr{$^\ddag$ \textit{School of Data Science \& Engineering, East China Normal University, Shanghai, China}}\\
       \affaddr{$^\S$ \textit{School of Computing, National University of Singapore, Singapore}}\\
       \affaddr{\textsf{zhou.xuan@outlook.com}}
}

\maketitle

\begin{abstract}
For performance reasons, conventional DBMSes adopt monolithic architectures. A monolithic design cripples the adaptability of a DBMS, making it difficult to customize, to meet particular requirements of different applications.
In this paper, we propose to completely separate the code of concurrency control (CC) from a monolithic DBMS. This allows us to add / remove functionalities or data structures to / from a DBMS easily, without concerning the issues of data consistency.
As the separation deprives the concurrency controller of the knowledge about data organization and processing, it may incur severe performance issues.
To minimize the performance loss, we devised a two-level CC mechanism. At the operational level, we propose a robust scheduler that guarantees to complete any data operation at a manageable cost. At the transactional level, the scheduler can utilize data semantics to achieve enhanced performance.
Extensive experiments were conducted to demonstrate the feasibility and effectiveness of our approach.
\end{abstract}

\vspace{-0.1cm}
\section{Introduction}
\vspace{-0.0cm}
Existing implementations of DBMSes are mostly monolithic. This goes against common practice of software engineering, where separation of concerns is an important principle.
Such monolithic design can be attributed to both tradition and performance consideration \cite{gray1981recovery,lomet2009unbundling}, which we believe are no longer valid in today's computing environment.
On the one hand, applications are diversifying. They impose increasingly diverse requirements on DBMS, in terms of both functionality and performance. To meet these requirements, application developers are increasingly incentivized to customize DBMSes, for instance, by adding new data types or indexing schemes.
On the other hand, hardware and platforms are evolving rapidly. We are constantly being forced to modify a DBMS to make the best of new hardware.
A monolithic design unavoidably makes a DBMS difficult to modify or customize.
We believe it is time to consider a loosely coupled architecture of DBMS, which is adaptable to diverse applications and platforms.

Attempts at DBMS decomposition dated back to two decades ago \cite{genesis,chaudhuri2000rethinking}, with limited progress and success.
It has been commonly accepted that a DBMS should be broken into several standard components, such as an interpreter, a query processor, a transaction manager, a storage manager, etc. However, existing DBMSes largely regard this decomposition as an explanatory breakdown instead of a guideline for modularization. Only in recent years, limited but concrete efforts for decomposing a DBMS have been visible. The Deuteronomy project of Microsoft \cite{lomet2009unbundling,lomet2009locking,unbundle_log,levandoski2015high} is a typical example, which attempted to decouple the transaction manager from the storage manager of a distributed database. Another example is today's ``big data'' platforms, such as Hadoop, which separates the data processor and the storage manager to achieve extensibility. Despite these efforts and their inspiring results, the answer to the problem of DBMS decomposition remains inconclusive.

Among all the coupling points in a DBMS, the one between the transaction manager and the data manager appears the most challenging to break \cite{hellerstein2007architecture}. In practice, it also causes the most pain to engineers who attempt to modify a DBMS. When adding a new data format or a new index to a DBMS, it is inevitable to also implement the transactional methods for the data format or index and ensure their compatibility with the entire system. When upgrading a transactional mechanism, such as adding a new concurrency control method, heavy modification has to be introduced to the code of data organization and processing.
To decompose a DBMS, it is crucial to separate the logic of transaction management from that of the data organization and processing component so that modifications on either component do not interfere with the other.

In this paper, we focus on Concurrency Control (CC), a major function of transaction management. We propose to completely separate CC from a DBMS, such that it becomes transparent to the rest of the system.
We call our approach Transparent Concurrency Control (TCC). While this separation is in theory possible, it does not come for free.
Once separated from the data layer, the CC layer is deprived of the knowledge about data semantics. This may introduce severe performance penalty.

A traditional DBMS performs CC at two levels -- the operational level and the transactional level. At the operational level, the CC mechanism ensures isolation among data operations, such as index lookup, index insertion, table scan, etc. To achieve efficiency, the CC methods are normally highly specialized for the particular data models and data processing programs \cite{kornacker1997concurrency}.
After the separation, such specialization is no longer possible, as the CC layer loses the knowledge about the data models or data processing methods. If we adopt a generic but blind CC mechanism, it is unlikely to perform well in all possible circumstances. We conducted experimental study to evaluate three generic CC mechanisms, 2PL, SSI and OCC, at the operational level. We found that the three mechanisms perform poorly on certain workloads, e.g., intensive index insertions.

The CC mechanism at the transactional level ensures the isolation among transactions. At this level, data semantics plays an important role. For instance, locking is widely used for isolation. However, after the separation, we cannot even determine the objects of locking, be it either a tuple or a table or a predicate, as such semantic objects are no longer visible to the CC layer.
Meanwhile, the semantic relationship between data operations is also missing. Traditional DBMSes often utilize these relationships to achieve improved performance.
For instance, as two insertions to the same table are semantically commutative, we can reorder the table insertions of different transactions to achieve a more efficient schedule.

This paper aims to tackle the TCC problem at the operational and transactional levels separately. At the operational level, we employ a try-and-error mechanism that can provide a certain guarantee about the efficiency of CC. At the transactional level, we provide interfaces for developers to declare data semantics to TCC, so that it can be utilized by the CC mechanism. We evaluated the two-level mechanism of TCC on the indexes of a real DBMS. The results demonstrate the potential of TCC in real-world implementation. It makes us optimistic about the feasibility to decompose a DBMS.

To summarize, we mainly made the following contributions in this paper:
\begin{enumerate}[nosep]
\item We introduced the concept and the architecture of TCC and proved its soundness (Sections~3 and 4).
\item We showed that separation of CC from DBMS will incur performance degradation. We identified two types of knowledge gaps, known as predictability gap and semantic gap, which are main reasons for such degradation (Section~5).
\item We devised a mechanism of TCC, which aims to bridge the two knowledge gaps at the operational and transactional levels respectively (Section~6). We conducted experiments to verify its effectiveness (Section~7).
\end{enumerate}

\vspace{-0.1cm}
\section{Related Work}
\vspace{-0.0cm}

There have been several attempts aiming at decomposing a DBMS into loosely coupled modules, with various purposes in minds.

In~\cite{chaudhuri2000rethinking}, Chaudhuri and Weikum envisioned a RISC-style system architecture, aiming to make a DBMS easier to tune and optimize. They propose to decompose a system coarsely into a storage manager and a query processor. Then the query processor can be further decomposed into an index manager, a SPJ query processor, an aggregator, etc. Such a decomposition is expected to enhance our ability of configuring and tuning a database, so as to improve its adaptability to changing workloads and environments. However, there has been little concrete follow-up research, and RISC-style DBMS remains a vision rather than a practical solution.

StagedDB~\cite{stagedb,harizopoulos2005stageddb} provides another approach to decompose a DBMS. It separates the workflow of query processing into a number of self-contained and connected stages, such as a parser, a query optimizer, a query executor, etc. Users are allowed to customize the stages, so that they can support user-defined data types, access methods or cost models~\cite{postgres_interpreter,genesis,exodus}. StagedDB aims at good performance of query processing. It does not address the modularity issue directly.

To the best of our knowledge, the Deuteronomy project of Microsoft~\cite{lomet2009unbundling,lomet2009locking,unbundle_log,levandoski2015high} is the most direct and recent effort to realize a decomposition of DBMS.
The architecture of Deuteronomy decomposes a database kernel into a Transaction Component (TC) responsible for concurrency control and recovery and several Data Components (DCs) responsible for data organization and manipulation.
Such an architecture allows system engineers to develop DCs independently, without concerning the work of TC.
As shown on the left of Figure~\ref{fig:arch}, this in effect places the transaction tier above data organization tier, which provides operational interfaces for data manipulation, such as retrieval, update, deletion and insertion of data items.
The downside of this architecture is two-fold.
First, DC is responsible for ensuring atomicity of data operations. This requires a built-in CC mechanism in the data organization tier. It means that CC has not been completely decoupled from the Data Component.
Second, DC must provide sufficient information for TC to detect conflicts among data operations.
The current implementation of Deuteronomy assumes that conflicts can be inferred through identifers of data objects.
However, in principle, conflicts are not necessarily inferrable from data identifers.
As shown in Figure~\ref{fig:coreference}, two seemingly separate items may refer to the same piece of physical data. If such implicit connection is unknown to TC, isolation is hardly achievable.
This assumption limits the flexibility of DC, as data sharing or co-referencing cannot be used freely.

\begin{figure}
\centering
\vspace{-0cm}
\includegraphics[width=0.48\textwidth]{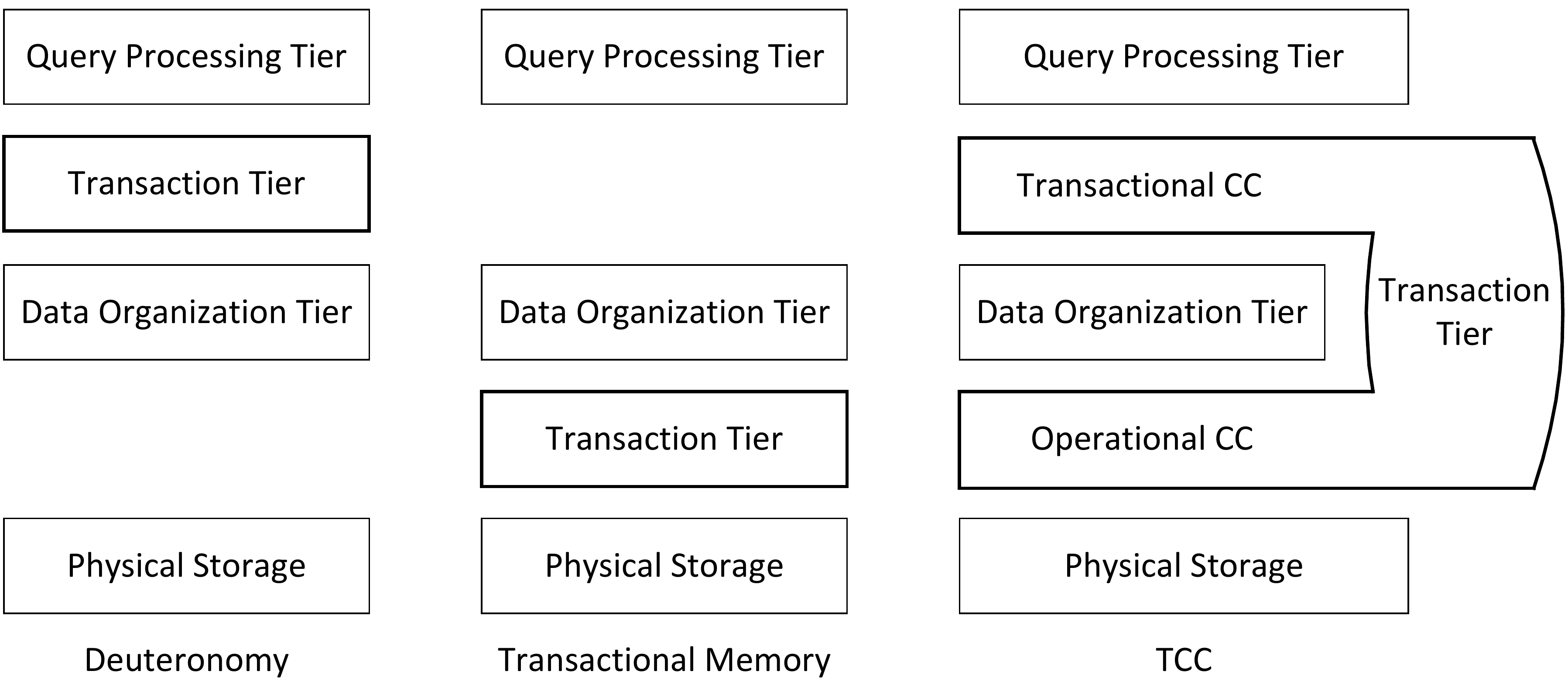}
\vspace{-0cm}
\caption{\small Possible Placements of the Transaction Tier}
\vspace{0cm}
\label{fig:arch}
\end{figure}

\begin{figure}
\centering
\includegraphics[width=.45\textwidth]{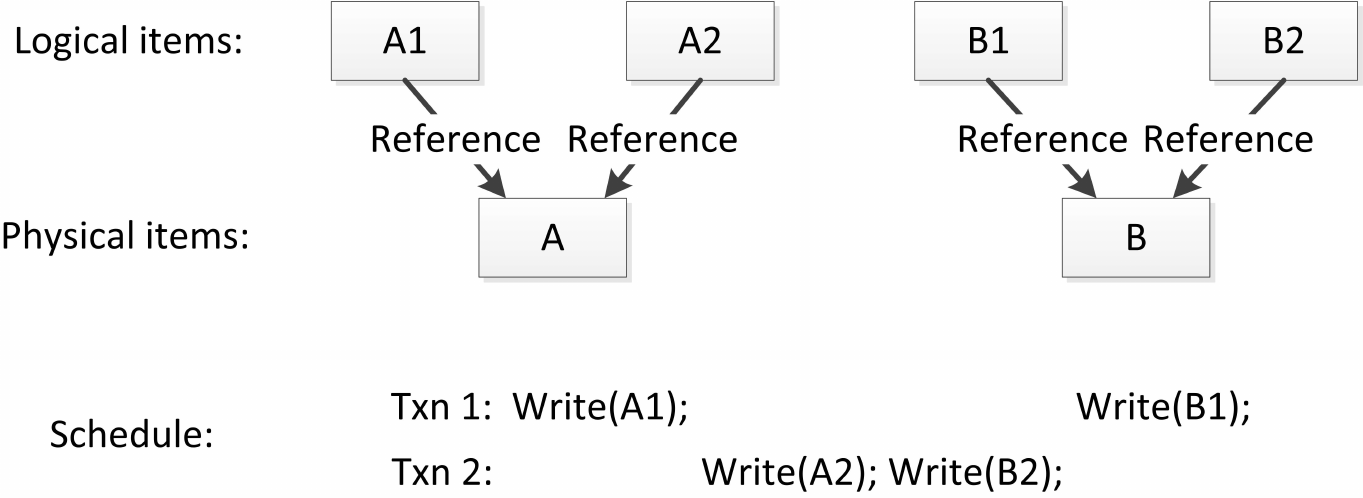}
\vspace{-0cm}
\caption{\small When logical items share physical data, serializablity cannot be ensured at the logical level alone. (As the transaction manager does not know that A1 and A2 or B1 and B2 refer to the same piece of data, it regards the above schedule serializable.)}
\vspace{-0cm}
\label{fig:coreference}
\end{figure}

By contrast, TCC expects to separate CC completely from the rest of the system.
As shown on the right of Figure~\ref{fig:arch}, TCC places an extra transaction tier between the data organization tier and the physical storage. This allows it to delegate the work of CC completely to the transaction tier.

It is not new to perform transaction management directly on the physical storage.
Transactional Memory (TM) is based on the same idea. TM provides transactional support on shared memory, in order to ease programmers' work on data synchronization. In recent years, TM has been a focus of intensive research \cite{herlihy1993transactional,cascaval2008software}, resulting in a number of hardware based and software based implementations (a.k.a. HTM and STM). Some recent work \cite{leis2014exploiting,cervini2015applying} has explored how to utilize HTM in database systems.
According to their study, due to the constraints imposed by hardware, HTM cannot be directly applied to database transactions. This limits its usage in a generic database system. STM is believed to incur high overheads \cite{cascaval2008software}, as it requires extra computation to perform concurrency control. In~\cite{osdi_06}, a ``transactional storage" was proposed to transactionalize block-addressable storage. However, the work is focused on the functionality of persistence and recovery.

The major issue faced by both HTM and STM is their lack of adaptability. TMs normally employ generic CC mechanisms, mostly OCC, which are not universally applicable to all programs of data manipulation. There are always corner cases~\cite{htm_corner_case}, in which they fail to perform. This is unacceptable to TCC. As TCC is supposed to be transparent, developers of the rest of the system should be allowed to implement any data manipulation method, without concerning any performance corner case. TCC deals with the adaptability issue through two approaches. On the one hand, its operational scheduler is able to learn from errors. This makes it eventually adaptable to any program of data manipulation. On the other hand, it provides interfaces for developers to input knowledge about data semantics, which can by utilized by its transactional scheduler to improve performance.

\vspace{-0.1cm}
\section{The Architecture}
\vspace{-0.0cm}
It is a common practice to decompose a database system into three tiers -- a query processing tier, a data organization tier and a physical storage tier~\cite{postgres_interpreter}. The query processing tier transforms a SQL query into a query plan and evaluates the plan by invoking relational operators, such as table scan, hash join, etc. The data organization tier is responsible for storing and maintaining structured data. It exposes interfaces of high-level data access to upper tiers, such as index lookup, tuple insertion, tuple update, etc. We call them \emph{data operations} or \emph{operations}. The physical storage tier exposes interfaces of low-level data access, such as read and write of data blocks. We call them \emph{r/w actions} or \emph{actions}.

In a traditional DBMS, the module of concurrency control is tightly integrated within the data organization tier.
Intuitively, the module functions at two levels. At the finer level, it schedules the actions enclosed in each data operation, to ensure atomicity of data operations. At the coarser level, it schedules the data operations, to enforce a certain level of isolation among transactions. For example, in MySQL, the implementation of B-tree involves both latches and locks~\cite{mohan_b_tree}. Latches enforce isolation among B-tree operations, such as lookup, insertion and deletion. Locks enforce isolation among transactions, each of which may involve multiple b-tree operations.

To separate the module of transaction management from the rest of the system, we are faced with three options. As Figure~\ref{fig:arch} illustrates, the first choice is to place the transaction tier above the data organization tier. This is the architecture adopted by Deuteronomy~\cite{lomet2009unbundling,lomet2009locking}. As mentioned earlier, in this architecture, the data organization tier itself will be responsible for performing CC among data operations.

The second choice is to place the transaction tier below the data organization tier. The transaction manager regards each transaction as a sequence of r/w actions on data blocks. If a DBMS relies on transactional memory / storage \cite{leis2014exploiting,htm_db} alone to implement its CC mechanism, it basically adopts this architecture. As this architecture enables a complete separation of the CC mechanism, we treat it as a baseline approach of TCC.
However, in this architecture, as the transaction tier lacks the knowledge about data organization, it is faced with severe performance issues. (Details about these issues will be elaborated in Section~\ref{sec:semantic_gap}.)

TCC adopts the third architecture (on the right of Figure~\ref{fig:arch}). It splits the transaction module into two tiers, and places one above and one below the data organization tier. We call the upper one the transactional CC tier and the lower one operational CC tier. They enforce isolation among transactions and data operations respectively.


As a result, the architecture of TCC consists of five tiers:

\textbf{Query Processing Tier:}
This tier interprets and executes
SQL queries. During the execution, it will invoke data operations offered by the data organization tier.

\textbf{Transactional CC Tier:}
This tier regards each transaction as a sequence of data operations, such as index lookup, tuple insertion, etc.
With the full knowledge about conflicts among data operations, it is able to schedule transactions to meet a desired isolation level, such as serializability.

\textbf{Data Organization Tier:}
This tier keeps the data organized in predefined structures, such as relational tables, B-tree indexes, etc.
It implements basic data operations, such as index lookup, tuple insertion, tuple update, table scan, etc.
In this tier, a data operation is further translated into a sequence of r/w actions on the physical storage.

\textbf{Operational CC Tier:}
This tier regards each data operation as a sequence of r/w actions, and employs a CC mechanism to ensure the serializability of data operations.

\textbf{Physical Storage Tier:} This tier executes r/w actions on the physical storage. In this paper, we assume that the database system uses block addressable storage. Therefore, the granularity of each r/w action is at the level of data blocks. We also assume that
each r/w action is atomic.
Should a DBMS employ a buffer manager to speedup data access, the buffer must
be located at this tier.

The interfaces exposed by the CC tiers are as follows:

\begin{enumerate}[nosep]
\item \emph{beginTx(int~tx\_id)}
This interface is invoked to start a transaction. The transaction has a unique identifier
\emph{tx\_id}.
The interface is provided by the transactional CC tier. It is supposed to be invoked by applications.

\item \emph{endTx(int~tx\_id)}
This interface is invoked to finish a transaction identified by \emph{tx\_id}.
It is also provided by the transactional CC tier and invoked by applications.
When a transaction ends, it either commits or aborts, depending on whether it violates the predefined isolation level.

\item \emph{abortTx(int~tx\_id)}
This interface is invoked by applications to abort a transaction identified by \emph{tx\_id}. It is provided by the transactional CC tier too.

\item \emph{beginOp(int~tx\_id, int~op\_id)}
This interface is provided by the operational CC tier.
It is invoked by the transactional CC tier before a data operation is invoked, to indicate the beginning of a data operation. We use \emph{tx\_id} to denote the identifier of the host transaction, and \emph{op\_id} to denote the identifier of the data operation.

\item \emph{endOp(int~tx\_id, int~op\_id)}
This interface is also provided by the operational CC tier.
This interface is invoked after a data operation finishes, to end the data operation identified by \emph{op\_id}. An operation may succeed or fail, depending on correctness of
its schedule.

\item \emph{read(int~tx\_id, int~op\_id, long~block\_id, char$*$~buf)}
The data organization tier invokes this interface to read the data block identified by \emph{block\_id}. Upon the invocation, the physical storage tier will copy the data in the block into the buffer \emph{buf} refers to.

\item \emph{write(int~tx\_id, int~op\_id, long~block\_id, char$*$~data)}
This interface is invoked to copy the \emph{data} into the block identified by \emph{block\_id} in the physical storage.
As calls of \emph{read} and \emph{write} all go through the operational CC tier, they are subject to the scheduling of the CC tier.
\end{enumerate}

Figure~\ref{fig:trans} illustrates the usage of the above interfaces. Suppose that the application submits a transaction to insert an entry into a table. Suppose that there is a B-tree index on the table. The application uses \emph{beginTx} and \emph{endTx} to specify the beginning and end of the transaction. The query processing tier transforms the SQL statement into two data operations in the data organization tier -- one inserts an entry into the B-tree and the other inserts a tuple into the table. The transactional CC tier encloses each data operation within a pair of \emph{beginOp} and \emph{endOp} calls. Between the two calls, the data organization tier invokes \emph{read} and \emph{write} interfaces to manipulate the data in the physical storage.

Such a design decouples CC from data organization tier completely. On the one hand, the CC tiers need not to care about how data is organized and processed. On the other hand, the data organization tier only needs to encapsulate data manipulation into data operations and invoke the read and write interfaces to access data in physical storage. It does not need to know the logic of CC mechanisms.

\begin{figure}
\centering
\includegraphics[width=.47\textwidth]{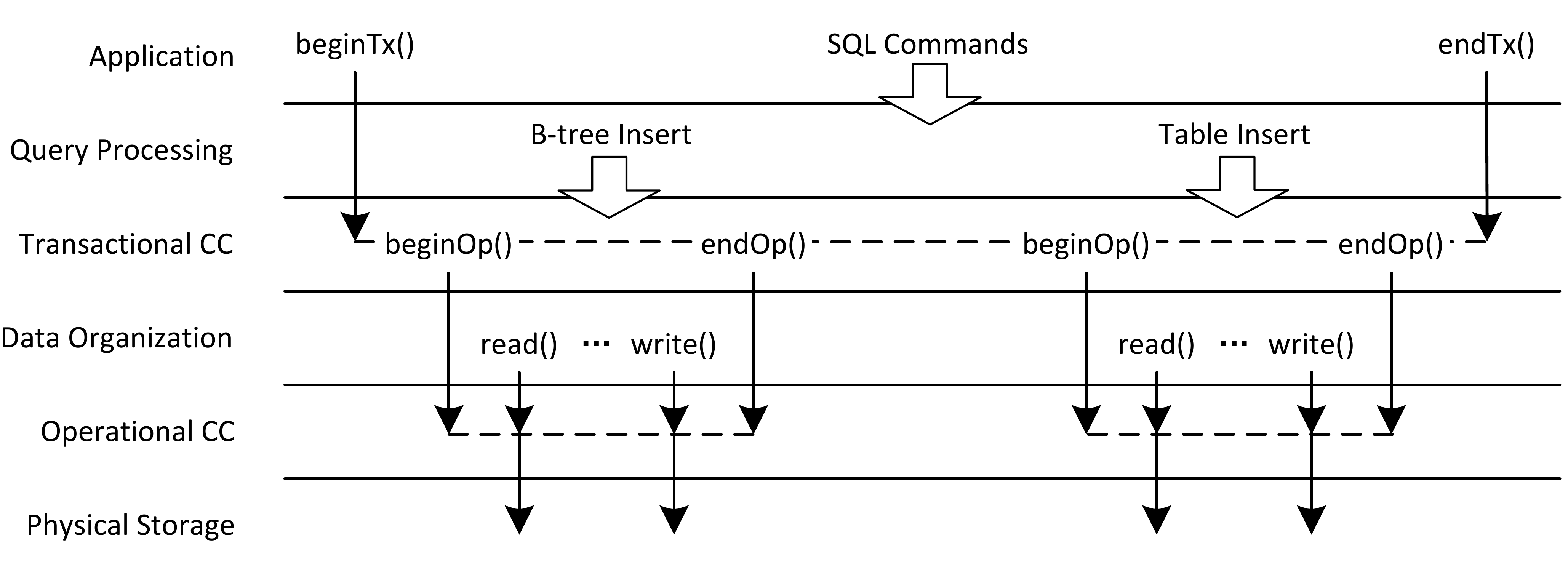}
\vspace{-0.2cm}
\caption{\small How the TCC Architecture Processes a Transaction}
\vspace{-0.2cm}
\label{fig:trans}
\end{figure}

A transaction module needs to deal with both concurrency control and recovery. In this paper, we focus on concurrency control. The function of recovery can be realized through a conventional page-level WAL mechanism. Due to space limitation, we do not further elaborate on it.

\vspace{-0.1cm}
\section{Correctness of TCC}
\vspace{-0.0cm}
\label{sec:correct}

In this paper, we consider only the isolation level of serializability. We show that TCC is able to enforce serializability.

\vspace{-0.1cm}
\subsection{Enforcement of Conflict Serializability}
\vspace{-0.0cm}

Conventional DBMSes treat serializability narrowly as \emph{conflict serializability}.
Enforcement of conflict serializability requires knowledge about conflicts among transactions.
As transactions are composed of data operations, it actually requires that
the CC layer should observe all conflicts among data operations.

Most textbooks on transaction management discuss only the conflicts among simple read and write operations. (By read and write operations, we refer to read and write of data objects rather than r/w actions on physical storage.) They create an illusion that conflict serializability can be enforced by simply locking data objects. In fact, data operations in real-world systems are of much higher complexity. Consider operations such as insertion/deletion of a data object, scan of an entire table, etc.
To capture conflicts among complex data operations, traditional DBMSes employ a variety of advanced locking mechanisms, such as key range locks, intention locks, predicate locks, etc.

Due to the separation, TCC 
is deprived of the options of using advanced locking mechanisms, such as predicate locks.
It has to infer conflicts among data operations based on their low-level actions on physical storage.
That is, it regards two data operations conflict, if and only if their r/w actions on the physical storage conflict. This approach greatly simplifies the CC mechanism. Meanwhile, it mandates the following prerequisite.

\begin{prerequisite}
\vspace{-0cm}
\label{pre:complete}
The information in the physical storage is complete and exclusive, such that the results of any sequence of data operations are exclusively determined by the state of the physical storage.
\vspace{-0cm}
\end{prerequisite}

Prerequisite~\ref{pre:complete} insists that all data and metadata should be stored in the physical storage. If any data or metadata is stored elsewhere, TCC may fail to capture the conflicts on this part of data. While this prerequisite appears trivial, system engineers 
must bear this prerequisite in mind, to prevent TCC from malfunctioning. For example, buffers must be placed within the physical storage layer, so that data accesses to the buffers are observable to TCC; data or metadata cannot be transmitted among data operations through shared variables, which TCC is unaware of.


\begin{theorem}
\vspace{-0cm}
\label{theorem:op_conflict}
Under Prerequisite~\ref{pre:complete}, two data operations conflict only if their r/w actions conflict.
\vspace{-0cm}
\end{theorem}
\begin{proof}
The proof is by contradiction. We assume that two data operations $o_1$ and $o_2$ conflict while their r/w operations do not conflict. Let $S_1$ and $S_2$ be the sequences of r/w operations of $o_1$ and $o_2$ respectively. As $o_1$ and $o_2$ conflict, there must be a sequence of operations $P$, such that the concatenated sequences $o_1o_2P$ and $o_2o_1P$ will yield different results. As $S_1$ and $S_2$ do not conflict, $S_1S_2$ and $S_2S_1$ must transfer the physical storage to the same state. Thus, we can conclude that the results of $P$ is not determined by the physical storage. This contradicts Prerequisite~\ref{pre:complete}.
\end{proof}

Theorem~\ref{theorem:op_conflict} states that TCC can capture all conflicts among data operations by observing the r/w actions.
This is sufficient for TCC to enforce conflict serializability.
In TCC, the operational CC tier is responsible for ensuring serializability among data operations,
and the transactional CC tier is responsible for ensuring serializability among transactions.
Generic CC mechanisms, such as 2PL, SSI and OCC, can be employed for the enforcement.

\vspace{-0.1cm}
\subsection{Beyond Conflict Serializability}
\vspace{-0.0cm}

Inferring operational conflicts at the physical level can be overkill.
In fact, when two r/w actions conflict on the physical storage, it is not necessary that their host data operations semantically conflict. For instance,
we can increment a counter twice, through two data operations. Physically the two operations conflict, as they modify the same piece of physical data. In effect, they do not, as they can be reordered without affecting the results.
As elaborated subsequently, 
conflict serializability at the level of physical storage will limit TCC's concurrency.
This issue is less serious to traditional DBMSes, as they detect conflict at the semantic level (the level of data objects), which helps them circumvent the worst cases.
To achieve good performance, TCC needs to go beyond conflict serializability.

In this paper, we consider View Serializability (VS), a less restrictive definition of serializability. As the traditional definition of VS considers only read and write data operations, we redefine it as follows, to make it applicable to general data operations.

\begin{definition}[View Equivalence]
\vspace{-0cm}
Two schedules $S$ and $S'$ of the same set of data operations are View Equivalent, if for all possible sequences of operations $A$ and $P$, the return values of the data operations in the concatenated sequence $ASP$ are identical to those in the sequence $AS'P$.
\vspace{-0cm}
\end{definition}

View equivalence requires not only that two schedules return the same results, but also that their subsequent operations (those of $P$) return the same results. That is, the two schedule should transform a database to the same state.
Two states of a database are semantically identical if they always return the same result to the same operation. They are not necessarily byte-to-byte identical in physical forms. For instance, in classical relational theory, two relational tables are equivalent, if they contain the same set of tuples, even though their tuples are stored in different orders.

\begin{definition}[View Serializability]
\vspace{-0cm}
Given a set of transactions $T$, a schedule $S$ is View Serializable, iff there exists a serial schedule $S'$ of $T$, such that $S$ and $S'$ are View Equivalent.
\vspace{-0cm}
\end{definition}

It is not difficult to prove that a conflict-serializable schedule is also view-serializable.
To harness the benefits of view serializability, TCC allows system developers to specify the conditions under which view serializability can be preserved, especially when conflict serializability is violated. For instance, the developer of B-tree can declare that two B-tree insertions are commutative, which means that the order of insertion has no impact on serializability. As a result, TCC no longer needs to consider the conflicts among B-tree insertions, even though they have modified the same 
data blocks.

\vspace{-0.1cm}
\section{Where does Performance Drop}
\vspace{-0.05cm}
\label{sec:two_issues}

Our goal is to optimize the performance of TCC, so that it can be an alternative to traditional CC mechanisms.

A performance issue one can easily think of is the granularity of CC.  As TCC operates at the block level, when data accesses are concentrated on a small number of blocks, the throughput may drop quickly.
In fact, this issue is not as serious as we expect.
In our experimental evaluation, we found that the granularity issue only occurs in a limited number of cases.
We leave the granularity issue to engineers of the data organization tier, who are supposed to keep hotspot data decentralized, and treat it as a principle of design. (This does not necessarily mean that we should sacrifice data locality. Hotspot data is a small amount of highly contended data. Even if we scatter the data on multiple blocks, they can still be accommodated by caches.)

A more serious challenge faced by TCC is information loss. Once the CC layer is separated from the rest of the system, the structures of data and the system's behaviorial patterns are no longer explicit to the CC mechanism. This may lead to serious performance degradation, as specialized designs cannot be adopted. We classify the issues of information loss into two categories -- \emph{predictability gap} and \emph{semantic gap}, and elaborate on them separately.

\vspace{-0.1cm}
\subsection{Predictability Gap}
\vspace{-0.0cm}
\label{sec:predict_gap}
There are limited types of data operations in a DBMS, which are repeatedly invoked to complete complex data processing.
As a result, there is a strong regularity in data accesses on the low-level storage. Such regularity has been utilized by traditional CC mechanisms to enhance performance.
For example, when performing B-tree insertion, if a leaf node is retrieved, it is guaranteed to be updated subsequently.
In MySQL, when a normal operation attempts to read a leaf node of a B-tree, it will place a shared latch on the node to allow more concurrency.
However, if the operation is a B-tree insertion, MySQL will place an exclusive latch on the leaf node upfront. This helps it avoid latch upgrade, which can easily lead to deadlocks (Figure~\ref{fig:btree_pattern}).

It is difficult for TCC to utilize such regular patterns in data accesses. When a B-tree insertion is reading a leaf node, TCC knows neither that it is a B-tree insertion nor that the block being accessed is a leaf node. It is then impossible for TCC to predict that there will be a follow-up modification. If TCC adopts a conventional CC mechanism, such as 2PL, B-tree insertion has to perform latch upgrade. As Figure~\ref{fig:btree_pattern} illustrates, if multiple B-tree insertions attempt to access the same leaf node concurrently, deadlock will be highly likely. To make the matter worse, if we retry the B-tree insertions whenever encountering a deadlock, it will incur more deadlocks or even starvation. The entire system may stop performing because of it.

\begin{figure}
\centering
\includegraphics[width=.35\textwidth]{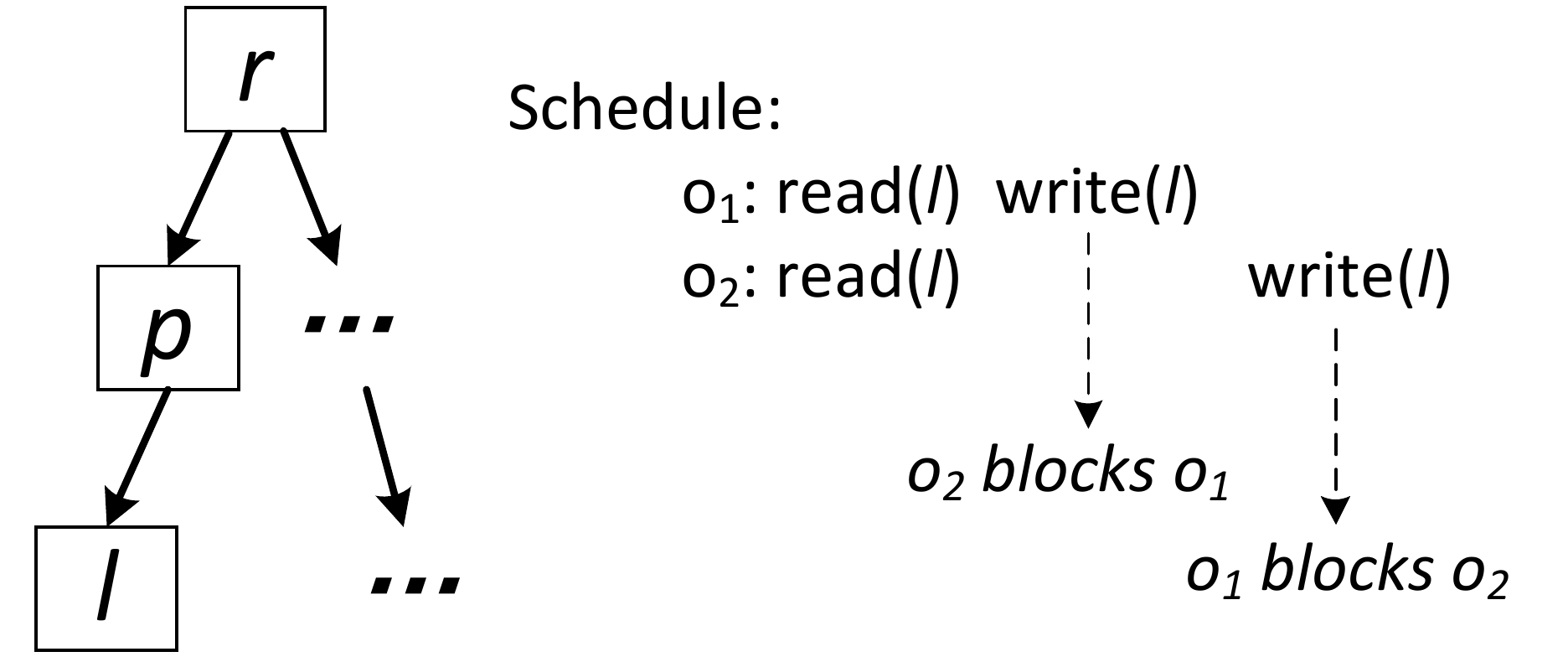}
\vspace{-0.2cm}
\caption{\small Data access sequences on B-tree that cause deadlock or abort.}
\label{fig:btree_pattern}
\vspace{-0.2cm}
\end{figure}

Without the knowledge about how each data operation works, TCC loses the ability to predict data operations' behaviors. Thus, it misses the opportunity to apply specialized mechanisms to improve the performance of CC.
We call this type of information loss ``predictability gap''.

To the best of our knowledge, all generic CC mechanisms that have existed suffer from predictability gap.
Figure~\ref{fig:general_pattern} illustrates a corner case no generic CC mechanisms can deal with, be it either 2PL or OCC. In this case, two concurrent data operations $o_1$ and $o_2$ update a sequence of data blocks in reverse orders. All generic CC mechanisms will allow $o_1$ and $o_2$ to update $p_1$ and $p_n$ concurrently.
This will surely lead to deadlock or abort. If $o_1$ and $o_2$ are invoked frequently, there will be performance degradation.
It is unacceptable that TCC be handicapped by such corner cases. However, we cannot resort to specialization, as we still need to hide the implementation details of data operations from TCC.
The only option left to us is to design a generic CC mechanism that is immune to predictability gap.

\begin{figure}
\centering
\includegraphics[width=.3\textwidth]{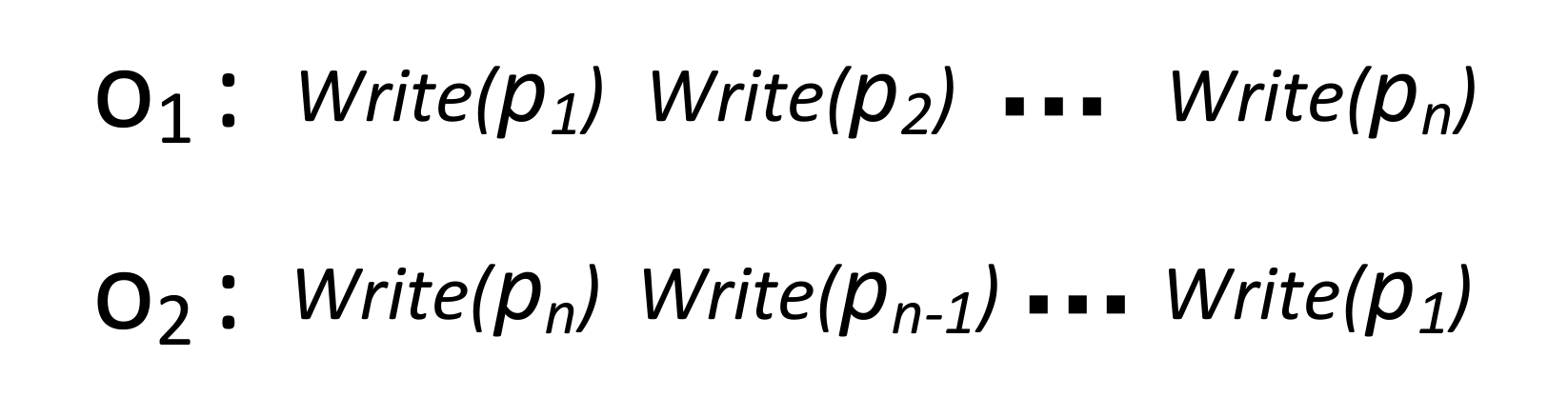}
\vspace{-0.2cm}
\caption{\small Data access sequences that embarrass all general-purpose CC mechanisms}
\label{fig:general_pattern}
\vspace{-0.4cm}
\end{figure}

In section~\ref{sec:op_level}, we will introduce a new CC mechanism, which can learn access patterns in a try-and-error manner. When performing or retrying a data operation, it acquires knowledge about its data access patterns. Then, it can utilize the knowledge in the subsequent retries. It proves to be robust against any corner cases.

\vspace{-0.1cm}
\subsection{Semantic Gap}
\vspace{-0.0cm}
\label{sec:semantic_gap}

We have mentioned that conflict serializability at the level of physical storage is too restrict for TCC to achieve good performance.
As a makeup, we introduced view serializability, which is based on the definition of view equivalence. View equivalence, in turn, is a semantic measure.
Its measurement requires the semantics of data operations, which we intend to hide from TCC.

For example, commutative operations and inverse operations~\cite{weikum2001transactional} are common semantics we can use to measure view equivalence. 
Suppose that transaction $T_1$ performs two B-tree insertions $o_1$ and $o_2$, and transaction $T_2$ performs one B-tree insertion $o_3$, all on the same leaf node $l$.
If we enforce conflict serializability restrictively, we can accept only two schedules of the operations, namely $[o_1, o_2, o_3]$ and $[o_3, o_1, o_2]$.
In fact, most real-world DBMSes accept the schedule $[o_1, o_3, o_2]$ too, simply because B-tree insertions are commutative. While it is possible that the two versions of $l$ resulted from $[o_1, o_2, o_3]$ and $[o_1, o_3, o_2]$ are not physically identical, they are view equivalent -- they are semantically identical to future data operations.

View serializability allows us to exploit more concurrency. However, TCC lacks the knowledge to judge view seraializability. This is known as ``semantic gap".

To deal with semantic gaps, we place a transactional CC tier atop the data organization tier.
It allows system engineers to explicitly declare semantic relationship between data operations (e.g., commutative operations, inverse operations).
Section~\ref{sec:tx_level} describes how TCC leverages data semantics to generate view serializable schedules.

\vspace{-0.1cm}
\section{The TCC Mechanism}
\vspace{-0.0cm}
\label{sec:solution}

The two-level architecture of TCC allows us to deal with the two information gaps separately. The operational tier deals with the predictability gap by adopting a try-and-error CC mechanism. To bridge the semantic gap, the transactional tier allows developers to declare semantic relationship among data operations.

\vspace{-0.1cm}
\subsection{Operational Scheduler}
\vspace{-0.0cm}
\label{sec:op_level}

Our scheduler at the operational level employs latching to enforce serializability of data operations.
The basic approach is two-phase latching -- an operation places latches when it is about to read or write a data block for the first time; it releases all the acquired latches after the operations. The scheduler fails an operation, if it suspects that it may violate serializability.
When an operation fails, the scheduler retries it immediately.
During a retry, it performs early latching to prevent the operation from failing again for the same reason. When an operation fails more, more early latches will be placed, so that the chance of a successful retry gradually increases. 

This try-and-error approach allows the scheduler to learn the behaviorial pattern of a data operation on the fly. As more retries are performed, the behavior of an operation becomes increasingly predictable. At a certain point, we can guarantee that the scheduler is able to complete the operation without further retry. To make this intuition work, we introduce the concept of progressiveness.

\begin{definition}[Progressiveness]
\vspace{-0cm}
\label{def:progress}
Let a data operation be a sequence of r/w actions. A scheduler is progressive if it can guarantee: whenever a data operation fails on an r/w action (i.e., the data operation is aborted because of a conflict on the action), the subsequent retries of the operation will not fail on the same r/w action again.
\vspace{-0cm}
\end{definition}

Progressiveness ensures that each r/w action of a data operation will fail at most once. If a data operation comprises $n$ r/w actions, it will fail at most $n$ times.
Therefore, a progressive scheduler guarantees to complete any data operation in a limited number of retries, no matter how complicated the situation is. Progressiveness means robustness.

To ensure progressiveness, the operational scheduler needs to think twice before deciding to fail an operation, as it cannot fail it on the same r/w action for more than once.
Suppose that two data operations $o_i$ and $o_j$ conflict. Then, there must be two r/w actions, $a_i$ of $o_i$ and $a_j$ of $o_j$, which attempt to access the same data block. Suppose that $a_i$ is ahead of $a_j$. We can distinguish among three types of situations:
\begin{itemize}[noitemsep]
\item[I.]  $o_i$ and $o_j$ have already failed on $a_i$ and $a_j$, in the previous attempts.
\item[II.] $o_i$ has never failed on $a_i$, while $o_j$ have failed on $a_j$.
\item[III.] $o_j$ has never failed on $a_j$.
\end{itemize}
To ensure progressiveness, in Situation I, we cannot abort either $o_i$ or $o_j$. In Situation II, we cannot abort $o_j$. In Situation III, it is always safer to abort $o_j$ rather than $o_i$.
Based on this observation, we come up with the following rules for our operational scheduler:
\begin{itemize}[nosep]
\item[1.] {\bf Basic Latching}. Whenever an operation $o$ conducts an r/w action $\langle p, m \rangle$ (where $p$ denotes the data block being accessed, and $m$ denotes the access mode, i.e., read or write), it is supposed to place a latch of mode $m$ on $p$. The latches will be held until $o$ succeeds or fails. This is basically two-phase latching, which ensures serializability among data operations.
\item[2.] {\bf Early Latching}. To deal with Situations I and II, we perform early latching.
    Whenever a data operation $o$ fails on an r/w action $\langle p, m \rangle$ for the first time, $o$ will record $\langle p, m \rangle$ in an immunity set $S_o$. When $o$ retries, it latches the blocks in its immunity set in advance. That is, for each $\langle p, m \rangle$ in $S_o$, $o$ will first place a latch of mode $m$ on $p$ before the execution starts.
    To avoid deadlocks in the early-latching phase: (1) we place latches in the order of block ids; (2) if a data operation will both read and write a block, we only place the \emph{write} latch.
    When early latching is in use, in Situations I and II, $o_i$ and $o_j$ will actually be executed in a serial order, as $o_j$ will be blocked by $o_i$ in the early-latching phase. Then, we can avoid aborting $o_i$ and $o_j$ on $a_i$ and $a_j$.
\item[3.] {\bf Early Abortion}. To deal with Situation III, we ensure that $o_j$, instead of $o_i$, is the one to abort.
    When a data operation $o$ performs an r/w action $\langle p, m \rangle$, if $o$ did not fail on the r/w action before, it will try to latch $p$ before the action. In this case, if another operation has already obtained the latch on $p$, instead of blocking $o$, we abort $o$ directly.
\end{itemize}

A scheduler following the above three rules will be deadlock free. Due to the use of early abortion, blocking can only occur in the early latching phase. As early latching is performed in a universal order, the aforementioned three rules alone cannot cause deadlock.
The following theorems confirm that our scheduler achieves serializability and progressiveness simultaneously.

\begin{theorem}
\label{theorem:op_schedule1}
If we perform scheduling by following Rules 1, 2 and 3, all data operations will be serializable.
\end{theorem}
\begin{proof}
The proof is by contradiction. If we assume that serializability does not hold, there must be a dependency cycle $o_1 \rightarrow o_2 \rightarrow \cdots \rightarrow o_n \rightarrow o_1$,  where $o_1, o_2, ..., o_n$ all complete successfully. For each dependency $o_i \rightarrow o_j$ in the cycle, we can conclude that it is not in Situation III. Otherwise, $o_j$ will abort. Then, $o_i \rightarrow o_j$ can only be in Situation I or Situation II. In either case, $o_j$ will not access any data until $o_i$ completes.
Then, there will be a deadlock among the operations $o_1 \cdots o_n$, as each operation is waiting for the preceding one to complete. Then, no operation can complete.
\end{proof}

\begin{theorem}
\label{theorem:op_schedule2}
If a scheduler follows Rules 1, 2 and 3 exactly (except the actions specified in Rules 1, 2 and 3, no other blocking or abortion is performed), it is a progressive scheduler.
\end{theorem}
\begin{proof}
First, if we apply Rules 1, 2 and 3, there will not be deadlock. To prove it, we assume that there is a deadlock in the form $o_1 \rightarrow o_2 \rightarrow \cdots \rightarrow o_n \rightarrow o_1$. We know that there is a universal order for early latching. Then, not all operations involved in the deadlock are in the early latching phase. Suppose that $o_j$ is not in the early latching phase and the r/w action blocking $o_j$ is $a_j$. Then, we can conclude that $o_j$ must have not failed on $a_j$. (Otherwise, $o_j$ should be blocked in the early latching phase.) According to Rule 3, if $o_j$ have not failed on $a_j$, $o_j$ should be aborted instead of being blocked. Then, the deadlock is impossible. We are in contradiction.

If deadlock is impossible, abort can only occur when we apply Rule 3. That is to say, a data operation can only fail on an r/w action where it has never failed. This is exactly what progressiveness needs.
\end{proof}

Algorithm~\ref{alg:op_exec} describes our scheduler. The duration of a data operation is divided into three phases. In the early latching phase, the operation latches all the blocks in the immunity set. During the execution phase, an operation performs updates only in its private workspace. This facilitates abortion -- to abort an operation, we simply discard its workspace. After the execution phase, the operation enters a clearing phase, in which it makes its modification visible to other operations.

In the scenario of intensive B-tree insertion (illustrated in Figure~\ref{fig:btree_pattern}), our progressive scheduler is superior to strict two-phase latching. Two concurrent B-tree insertions may conflict when they attempt to upgrade their latches on the same leaf node $l$. In this case, our scheduler aborts both insertions, and adds the r/w action $\langle l, write \rangle$ to their immunity sets. When it retries the two B-tree insertions, it will place a $write$ latch on $l$ at the very beginning. This guarantees the success of the retries. If we employ strict two-phase latching, the two B-tree insertions may fail repeatedly.

Compared to traditional optimistic CC mechanisms, such as OCC and SSI, early latching may seem too pessimistic. In fact, our basic assumption is that data operations are all short. In real-world systems, this assumption is valid, since long and sophisticated data manipulations are always composed of short and generic operations. Under this assumption, it is unlikely that early latching will hurt performance severely. It is more important to ensure the progressiveness of operation execution, as it frees system developers from the concerns on performance corner cases.
In contrast to operations, lengths of transactions are less controllable, as they are determined by applications. This is the reason why we decide not to apply the same progressive scheduler to the transactional level.

\begin{algorithm2e}[t]
\begin{scriptsize}
\caption{\small The Processing of a Data Operation}
\label{alg:op_exec}
\SetAlgoLined
\SetKwFunction{FbeginOp}{beginOp}
\SetKwFunction{Fread}{read}
\SetKwFunction{Fwrite}{write}
\SetKwFunction{FendOp}{endOp}
\SetKwProg{Fn}{Function}{:}{end}
\Fn{\FbeginOp{$t$,$o$}}{
    // Start of the Early Latching Phase\\
    sort $o$'s immunity set $S_o$ based on block ids \\ \label{alg:op:1}
    remove any $\langle p, read \rangle$ from $S_o$, if $\langle p, write \rangle \in S_o$\\
    \For {each $\langle p, m \rangle \in S_o$} {
        $o$ places a latch of mode $m$ on $p$\\
        // $o$ will be blocked, if $p$ has already been latched
    } \label{alg:op:2}
    // Start of the Execution Phase\\
}
\Fn{\Fread{$t$,$o$,$p$,$buf$}}{
        \If {$o$ has not latched $p$} {
            $o$ places a read latch on $p$ \label{alg:op:3}\\
             \If {$o$ is blocked on Line~\ref{alg:op:3}} {
            set $S_o = S_o \cup \{\langle p, read \rangle\}$ \label{alg:op:31} \\
            fail $o$ \\
            \FendOp{$t$,$o$}\\
            \KwRet \label{alg:op:32}
            }
        }
        read the block identified by $p$ into $buf$\\
}
\Fn{\Fwrite{$t$,$o$,$p$,$buf$}}{
        \If {$o$ has not latched $p$} {
            $o$ places a write latch on $p$ \label{alg:op:4}\\
            \If {$o$ is blocked on Line~\ref{alg:op:4}} {
            set $S_o = S_o \cup \{\langle p, write \rangle\}$ \label{alg:op:41}\\
            fail $o$ \\
            \FendOp{$t$,$o$}\\
            \KwRet \label{alg:op:42}
            }
        }
        write $buf$ into the block identified by $p$\\
}
\Fn{\FendOp{$t$,$o$}}{
    // Start of the Clearing Phase\\
    \If{$o$ failed}{
        $o$ unlatches all acquired latches\\
        \KwRet
    }
    \If {any data accessed by $o$ is uncommitted} { \label{alg:op:20}
        abort $t$ \\
        \KwRet
    } \label{alg:op:21}
    make $o$'s modification visible\\ \label{alg:op:7}
    \For {each block $p$ accessed by $o$} {
        set $o.latchcount[p] = p.latchcount$\\
        set $p.latchcount = p.latchcount+1$\\
        $o$ unlatches $p$ \label{alg:op:5}\\
    }
    // Start of the Locking Phase of the Transactional Scheduler \\
    \For {each block $p$ accessed by $o$} { \label{alg:op:6}
        set $m$ to shared mode if $o$ has read $p$ \\
        set $m$ to exclusive mode if $o$ has modified $p$ \\
        $t$ places a lock of mode $m$ on $p$ \label{alg:op:8}\\
        \uIf {$o.latchcount[p] > p.lockcount$ } { \label{alg:op:9}
            add $o.latchcount[p]$ to $p.incre$\\
            abort $t$ \label{alg:op:10}
        }\Else{
            set $p.lockcount = p.lockcount+1$\\
            \While{ $p.lockcount \in p.incre$}{
                remove $p.lockcount$ from $p.incre$\\
                $p.lockcount = p.lockcount+1$
            }
        }
    }
}
\end{scriptsize}
\end{algorithm2e}

\vspace{-0.1cm}
\subsection{Transactional Scheduler}
\vspace{-0.0cm}
\label{sec:tx_level}
The operational scheduler ensures a serial order of data operations. The transactional scheduler is supposed to schedule the operations to enforce serializability among transactions.
In theory, it can employ any CC mechanism to enforce serializability, including 2PL, SSI, OCC, etc.
However, there is a distinction between TCC and traditional DBMSes in transactional scheduling. In traditional DBMSes, the scheduler can predict conflicts between data operations prior to their execution, by comparing object ids or query predicates.
In TCC, the scheduler can only observe conflicts during or after the execution of operations, as confliction can only be inferred from r/w actions on the physical storage (Section~\ref{sec:correct}).
This makes the design of the transactional scheduler less straight forward.

We devised two transactional schedulers for TCC -- a basic scheduler which applies two-phase locking to enforce conflict serializability, and an extended scheduler which can relax the schedules to view serializability.

\vspace{-0.1cm}
\subsubsection{The Basic Scheduler}
\vspace{-0.0cm}
To perform 2PL, we need to determine the objects of locking. The locking objects of traditional DBMSes, such as tuple, table and predicate, do not apply, as they are unknown to TCC. Therefore, TCC has to place locks directly on data blocks. As mentioned earlier, r/w actions on data blocks enable TCC to capture all conflicts among data operations. Locking blocks suffices to achieve 2PL.

Our design of the 2PL mechanism has to consider the particular situation of TCC. First, we decide to perform locking only after a data operation completes.
If we perform locking during the execution of a data operation, it will interfere with the work of the operational scheduler, making progressiveness difficult to achieve.
As shown in the \emph{endOp} function of Algorithm~\ref{alg:op_exec}, we perform locking after
the Clearing Phase of each data operation. More precisely, locks are added after all the latches are released.
Separating latching and locking phases enables us to avoid unresolvable deadlocks.
If we perform locking before latches are released, latches and locks may together constitute a deadlock. Such deadlocks are expensive to detect and resolve.
For example, in Figure~\ref{fig:lock_latch}, two transactions $T_1$ and $T_2$ are executed concurrently. At the beginning, $T_1$ executes an operation $o_1$ to update the block $x$. It thus holds a lock on $x$. Then, $T_2$ executes an operation $o_3$ to update the blocks $x$ and $y$. When $T_2$ attempts to lock $x$, it is blocked by $T_1$, while holding latches on both $x$ and $y$. If $T_1$ then executes an operation $o_2$ that updates $y$, it has to wait for $T_2$'s latch on $y$. As a result, a deadlock is formed.

\begin{figure}[t]
\centering
\includegraphics[width=.3\textwidth]{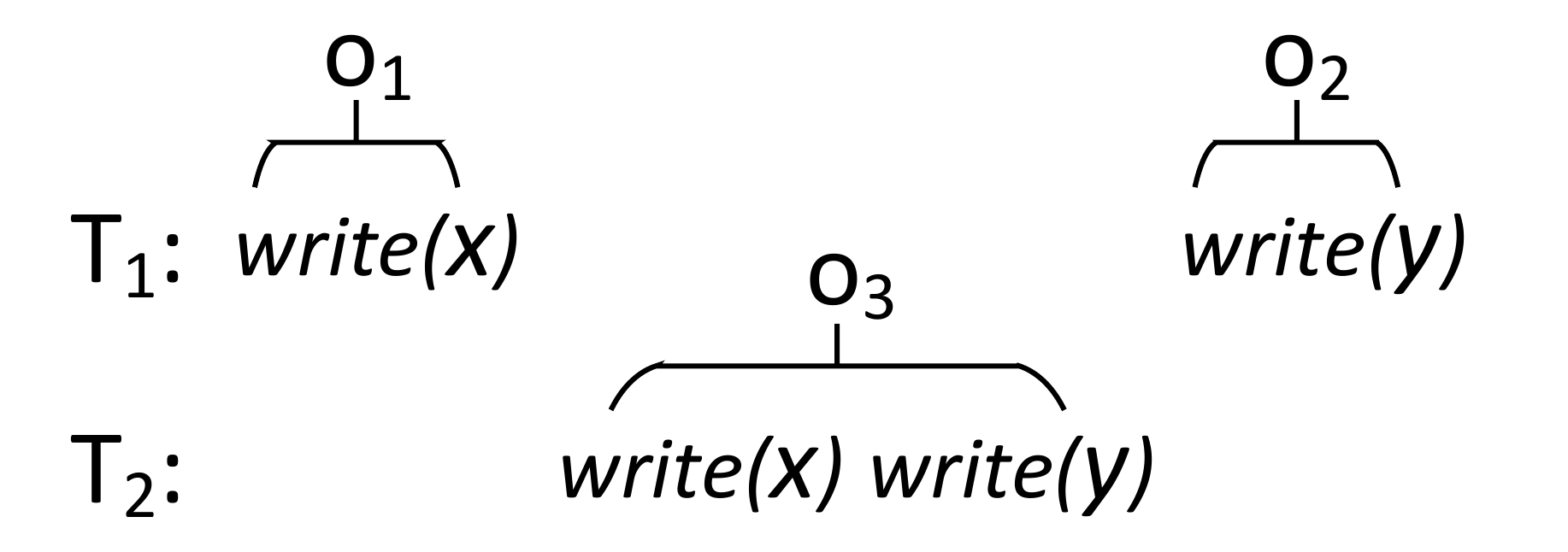}
\vspace{-0.1cm}
\caption{\small An example where locks and latches form a deadlock.}
\label{fig:lock_latch}
\vspace{-0.4cm}
\end{figure}

Second, since the locking phase is separated from the latching phase, we must guarantee that transactions place locks in the same order as their data operations place latches.
That is to say, if two data operations conflict, resulting in a dependency $o_i \rightarrow o_j$, then the transaction of $o_i$ must place the lock before the transaction of $o_j$ does.
To ensure the consistency between latching and locking orders, whenever a transaction obtains a lock, we check if its locking order complies with the latching order. If it does not, we abort the transaction (Line~\ref{alg:op:10} of Algorithm~\ref{alg:op_exec}). We maintain a latch counter and a lock counter for each data block, which will be incremented during the latching and locking phases respectively. If a transaction performs locking in the right order, it is supposed to observe identical latch and lock counters. If there is a gap between the two counters (Line~\ref{alg:op:9} of Algorithm~\ref{alg:op_exec}), it means that the locking order and the latching order are inconsistent.

A possible concern is that the separation between the latching and locking phases may lead to high abort rate. According to our experiment study (Section~\ref{exp:tx_scheduler}), this is unlikely, as the interval between the two phases is sufficiently small.

Recoverability refers to the ability to abort transactions correctly.
When a transaction aborts, it needs to perform extra writes on the data blocks it has modified, to recover them to the original versions.
It has been proven that recoverability is achievable if we disallow access on uncommitted data~\cite{weikum2001transactional}.
In principle, 2PL guarantees that no uncommitted data is accessed by any transaction.
As to TCC, since it performs locking after latches are released,
it is possible that a data operation accesses uncommitted data.
To ensure recoverability, we simply abort transactions that accessed uncommitted data (Line~\ref{alg:op:20}-\ref{alg:op:21} of Algorithm~\ref{alg:op_exec}).

\begin{algorithm2e}
\begin{scriptsize}
\caption{\small The Processing of a Transaction}
\label{alg:tx_exec}
\SetAlgoLined
\SetKwFunction{FbeginTx}{beginTx}
\SetKwFunction{FabortTx}{abortTx}
\SetKwFunction{FendTx}{endTx}
\SetKwProg{Fn}{Function}{:}{end}
\Fn{\FbeginTx{$t$}}{
    initialize $t$
}
\Fn{\FabortTx{$t$}}{
    sort the operations $O_t$ of $t$ by the reverse order of their invocation\\
    \For{each operation $o \in O_t$}{
        \uIf{$o$ has an inverse operation $o^{-1}$}{
            $t$ invokes $o^{-1}$
        }\Else{
            undo $o$ through the undo log
        }
    }
    release all the locks of $t$
}
\Fn{\FendTx{$t$}}{
    release all the locks of $t$
}
\end{scriptsize}
\end{algorithm2e}

TCC provides two ways to rollback a transaction. First, it maintains undo logs and uses them to recover data blocks to older versions. As an aborted transaction has already locked the data it has modified, no other transaction can access the data before the rollback is finished. Second, system engineers may have created inverse operations for some data operations. Then, we can cancel a data operation by executing its inverse operation. The details will be discussed in the Section~\ref{sec:expanded_schedule}.

Algorithm~\ref{alg:tx_exec} depicts how the basic transactional scheduler works.
It is worth noting that our transactional scheduler is not deadlock free. It thus requires a deadlock detector. Moreover, our transactional scheduler does not ensure progressiveness. Since our progressive scheduler can be overly pessimistic, applying it to the transaction level may hurt the concurrency of long-duration transactions. We consider it as application developers' responsibility to ensure the performance of transactions. This is how the state-of-the-art software development works.

\vspace{-0.1cm}
\subsubsection{The Extended Scheduler}
\vspace{-0.0cm}
\label{sec:expanded_schedule}

The basic scheduler enforces conflict serializability. As discussed previously, conflict serializability can be overkill.
To improve the concurrency of transaction processing, we have introduced the concept of view serializability, which allows us to take data semantics into consideration.

An important type of data semantics is commutativity.

\begin{definition}[Commutative Operation]
\vspace{-0cm}
\label{def:commutable}
Two operations $o_i$ and $o_j$ are commutative, iff for any two sequences of data operations, say $\alpha$ and $\beta$, the two schedules $[\alpha, o_i, o_j, \beta]$ and $[\alpha, o_j, o_i, \beta]$ are view equivalent.
\vspace{-0cm}
\end{definition}

TCC provides an interface \emph{add\_commutativity(int $op_1$, void *$args_1$, int $op_2$, void *$args_2$)} for system developers to declare that data operations of type \emph{op$_1$} and \emph{op$_2$} are commutative operations.
$args_1$ and $args_2$ are the argument lists of $op_1$ and $op_2$ respectively. They are used to specify the conditions where commutativity holds.
For example, suppose that the type of B-tree insertions is identified by \emph{$1$}. A developer can invoke \emph{add\_commutativity(1, null, 1, null)} to notify TCC that B-tree insertions are always mutually commutative.

Conflicts among consecutive commutative operations can be ignored when we enforce view serializability. This can be confirmed by the following theorem.

\begin{theorem}
\vspace{-0.2cm}
\label{thm:final}
A schedule preserves view serializability if the following conditions are satisfied:
\begin{itemize}[nosep]
\item Suppose $D$ is the complete set of dependencies among the transactions.
\item Suppose $D'$ is the complete set of dependencies caused by consecutive commutable operations.
\item The dependency graph $G$ consisting of $D-D'$ is acyclic.
\end{itemize}
\end{theorem}
\begin{proof}
For each pair of dependency $T_p \rightarrow T_q \in D'$, we can rearrange the order of $T_p$ and $T_q$, i.e., turning it to $T_q \rightarrow T_p$, without violating view serializability.

If the schedule does not satisfy view serializability, there must be a dependency cycle. Then, the cycle must not contain a dependency in $D'$. Otherwise, we can rearrange the dependency to break the cycle.
\end{proof}

To take advantage of commutativity in TCC, we extend the basic scheduler.
We regard locks hold by commutative data operations compatible. For example, if transaction $T_1$ executed a B-tree insertion and modified the leaf node $l$, $T_1$ will hold an exclusive lock on $l$. Then, when another transaction $T_2$ executes a B-tree insertion and modifies the same leaf node $l$, $T_2$ can be granted with an exclusive lock on $l$ too. (To the basic scheduler, $T_2$ is supposed to be blocked.) This preserves view serializability, as the execution order of commutative operations can be arbitrary.

However, when commutativity is considered, extra measures are required to ensure recoverability.
As an exclusive lock is no longer exclusive to commutable operations, a transaction may read uncommitted data. If the uncommitted data is aborted, we have to perform cascading abort, which will be expensive.
While we can forbid access on uncommitted data, it makes commutativity useless.
To undo commutative data operations, the best strategy is to use \emph{inverse operations}.

\begin{definition}[Inverse Operation]
\vspace{-0cm}
\label{def:inverse}
$o^{-1}$ is an inverse operation of the operation $o$, iff for any two sequences of data operations, say $\alpha$ and $\beta$, the two schedules $[\alpha, \beta]$ and $[\alpha, o, o^{-1}, \beta]$ are view equivalent.
\vspace{-0cm}
\end{definition}

TCC provides an interface \emph{addInverse(int $op_1$, void *$args_1$, int $op_2$, void *$args_2$)} for system developers to declare inverse operations. This interface specifies that $op_2$ is an inverse operation of $op_1$.
$args_1$ and $args_2$ are the argument lists of $op_1$ and $op_2$ respectively.
For example, suppose that B-tree deletion is an inverse operation of B-tree insertion.
We can declare the inverse operations by invoking \emph{addInverse(btreeInsert, [$key$, $value$], btreeDelete, [$key$])}. It indicates that if we perform B-tree deletion on $key$, it will undo the B-tree insertion with the same $key$.

If a data operation's uncommitted data has been accessed by its commutative operations, we can abort it by simply invoking its inverse operation, without also aborting its commutative operations. The following theorem justifies this.

\begin{theorem}
\vspace{-0cm}
\label{theorem:inverse}
Suppose the operations $o$ and $o'$ are commutative, and $o^{-1}$ is an inverse operation of $o$. Given any two sequences of data operations, say $\alpha$ and $\beta$,  $[\alpha, o, o', o^{-1}, \beta]$ and  $[\alpha, o', \beta]$ are view equivalent.
\vspace{-0cm}
\end{theorem}
\begin{proof}
The proof is straightforward. By Definition~\ref{def:commutable}, $[\alpha, o, o', o^{-1}, \beta]$ and $[\alpha, o', o, o^{-1}, \beta]$ are view equivalent. By Definition~\ref{def:inverse}, $[\alpha, o', o, o^{-1}, \beta]$ and $[\alpha, o', \beta]$ are view equivalent. Thus, $[\alpha, o, o', o^{-1}, \beta]$ and $[\alpha, o', \beta]$ are view equivalent.
\end{proof}

When we abort a transaction, we undo its operations serially in reverse order. For an operation that is not commutative with any other operations, we undo it through the undo log. For an operation that has commutative operations, we invoke its inverse operation to undo it. Different from executing an undo log, an inverse operation can possibly be blocked by other transactions. In this case, instead of letting it be blocked, we choose to fail the inverse operation and retry it.
And we repeat retrying until it succeeds.

In this paper, we consider only commutative and inverse operations. It is possible to define and exploit other types of data semantics in TCC. However, this is not within the scope of our current work.

\vspace{-0.1cm}
\section{Experimental Study}
\vspace{-0.0cm}

To evaluate the practicality of TCC, the best way is to apply TCC to an existing database system, whose design is completely oblivious to how TCC works.
The purpose of TCC is to make concurrency control transparent to database engineers. If we create a new database system based on TCC, we will be inclined to tailor its design to the particular mechanisms of TCC. This will make the evaluation less objective.
However, a complete substitution of the existing CC mechanism in a DBMS is extremely costly, if not impossible. The code of CC is usually intertwined with a large number of components of a DBMS, including the metadata manager, the storage space manager, the table manager, the indexer, etc. A complete deployment of TCC requires us to re-engineer all the components. It is beyond the capability of our research team.
As a compromise, we chose to apply TCC to only the indexes of a DBMS. Indexes are typical data structures in data management. Their concurrency controllers are usually highly specialized. In the TCC architecture, they are likely to be affected by the predictability and semantic gaps. Therefore, evaluation on indexes can show how well TCC deals with the two gaps in a generic DBMS.

\vspace{-0.1cm}
\subsection{The Implementation}
\vspace{-0.0cm}
\label{sec:impl}

Our codebase is Shore-MT~\cite{johnson2009shore}, a well used research prototype of RDBMS. It adopts 2PL for transaction-level CC and applies specialized CC mechanisms to indexes and metadata.

B-tree is the only type of index used by Shore-MT. We disabled the original concurrency controller on the B-trees of Shore-MT, and supplemented it with the TCC mechanism.
Shore-MT's B-tree are disk-resident. Any access to a B-tree node needs to first fix the underlying block in the buffer to avoid invalid access. Therefore, we regarded the ``fix'' routines as the read/write interface of the physical storage, and deployed the TCC module around it. This allows TCC to capture every r/w action on B-tree.

We implemented four mechanisms of TCC. The first three adopt the architecture of transactional memory (i.e., the middle one in Figure~\ref{fig:arch}) and apply standard 2PL, SSI and OCC respectively to enforce serializability.
These three mechanisms ignore the existence of data operations, and simply treat each transaction as a sequence of r/w actions. They may thus suffer from the the predictability and semantic gaps introduced in Section~\ref{sec:two_issues}.
We denote them by $TCC_{2PL}$, $TCC_{SSI}$ and $TCC_{OCC}$ respectively.
The fourth one is the TCC mechanism we proposed in this paper (adopting the architecture on the right of Figure~\ref{fig:arch}). We denote it by $TCC$. Since $TCC$ uses two transactional schedulers, a basic one and an extended one (Section~\ref{sec:tx_level}), we denote a variant of $TCC$ that uses only the basic transactional scheduler as $TCC_{basic}$.

To preserve the ACID of transactions, we need to integrate the CC of B-trees with that of the rest of Shore-MT. For $TCC$, we let its transactional scheduler and the rest of Shore-MT share the same lock manager.
We did the same to $TCC_{2PL}$. For $TCC_{SSI}$ and $TCC_{OCC}$, we implemented two variants of Shore-MT, $MT_{SSI}$ and $MT_{OCC}$, which uses SSI and OCC for concurrency control. Then, we integrated the schedulers of
$TCC_{SSI}$ and $TCC_{OCC}$ into those of $MT_{SSI}$ and $MT_{OCC}$ respectively.

Shore-MT does not support MVCC.
To implement $TCC_{SSI}$ and $MT_{SSI}$, we carved out an additional storage space to store old versions of data.
All versions of a data block are linked together, so that a transaction can easily retrieve the proper version to read.
Regarding the implementation of $TCC_{OCC}$ and $MT_{OCC}$, we maintain a write set and a read set for each transaction. During the validation stage, a transction locks the write set and validates the read set.

\vspace{-0.1cm}
\subsection{Experiment Setup}
\vspace{-0.0cm}

We compared TCC against the original CC mechanisms of Shore-MT. We had three versions of Shore-MT, $MT_{2PL}$, $MT_{SSI}$ and $MT_{OCC}$.
$MT_{2PL}$ is the original Shore-MT, which uses 2PL for concurrency control. To achieve its best performance, we applied two of its optimization patches, i.e., Speculative Lock Inheritance (SLI)~\cite{johnson2009improving} and Early Lock Release (ELR)~\cite{johnson2010aether}. $MT_{SSI}$ and $MT_{OCC}$ are variants of Shore-MT that uses SSI and OCC for concurrency control. They were implemented to cooperate with $TCC_{SSI}$ and $TCC_{OCC}$.

The experiments were carried out on an HP workstation equipped with $4$ Intel Xeon E7-4830 CPUs (with 32 cores and 64 physical threads in total) and a SATA-2T disk. The operating system was 64-bit Ubuntu 12.04. In most of the experiments, we set the buffer size to $32$ MB. For the experiments on TPC-C, we set the buffer size to $4$ GB (default setting of ShoreKit). For the experiments on TATP, we set the buffer size to $1$ GB (default setting of ShoreKit). We intentionally kept the buffer size large, to minimize I/O wait time. This helps to maximize concurrency control's influence on performance. For the same reason, we turned off the logging of Shore-MT.

\vspace{-0.1cm}
\subsection{Experiments on Operational Scheduler}
\vspace{-0.0cm}
\label{sec:exp:op_scheduler}

Our operational scheduler was designed to bridge the predictability gap. It is supposed to handle any data operation efficiently, regardless of its data access patterns.
To evaluate the robustness of our operational scheduler, we performed experiments on a variety of scenarios, including different cases of B-tree insertion and an artificial corner case (such as the one depicted in Figure~\ref{fig:general_pattern}).

In the experiments on B-tree insertion, we created a B-tree index of $10^6$ records and ran two types of workload on it. In the first type of workload, each transaction contains a single tuple insertion, which inserts a tuple into a randomly selected leaf node of the B-tree. It represents the case of low contention. In the second type of workload, the transaction performs sequential tuple insertion, such that all transactions contend for the last leaf node. It represents the case of high contention. 

\begin{figure*}
  \centering
  \begin{minipage}{0.5\linewidth}
      \begin{figure}[H]
          \subfigure[\small Throughput]
            {
                \includegraphics[width = 0.49\linewidth]{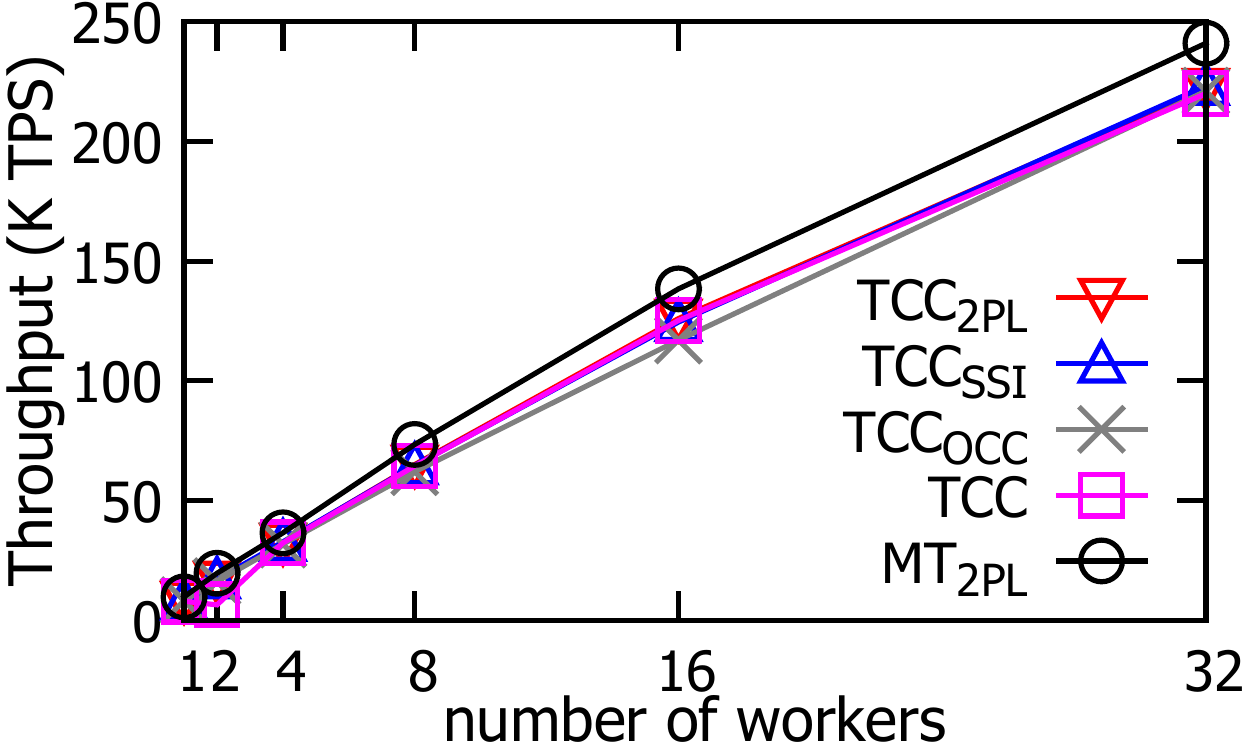}
                \label{fig:random_insert_throughput}
            }
            \hspace{-0.5cm}
            \subfigure[\small Abort Rate]
            {
                \includegraphics[width=0.49\linewidth]{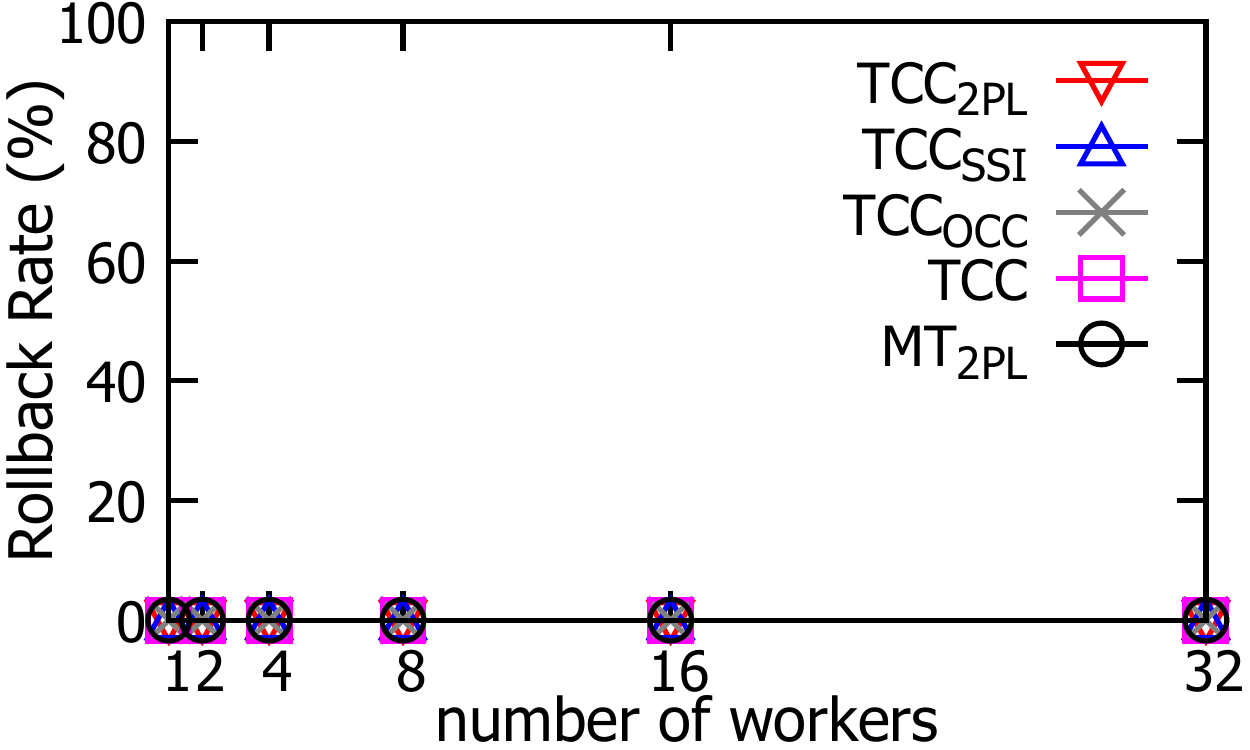}
                \label{fig:random_insert_abort}
            }
            \vspace{-0.5cm}
            \caption*{\small I. Random Insertion.}
            \vspace{-0.2cm}
      \end{figure}
  \end{minipage}
  \hspace{-0.5cm}
  \begin{minipage}{0.5\linewidth}
      \begin{figure}[H]
         \subfigure[\small Throughput]
             {
                \includegraphics[width = 0.49\linewidth]{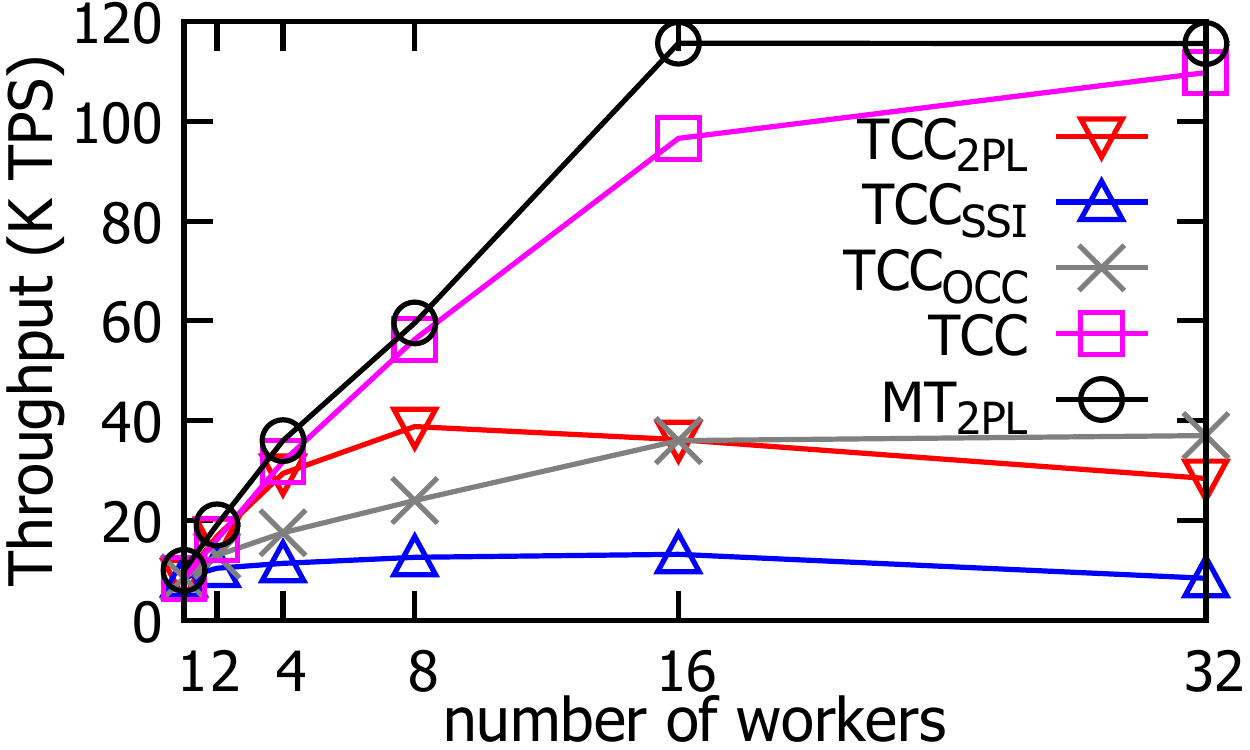}
                \label{fig:sequential_insert_throughput}
            }
            \hspace{-0.5cm}
            \subfigure[\small Abort Rate]
            {
                \includegraphics[width=0.49\linewidth]{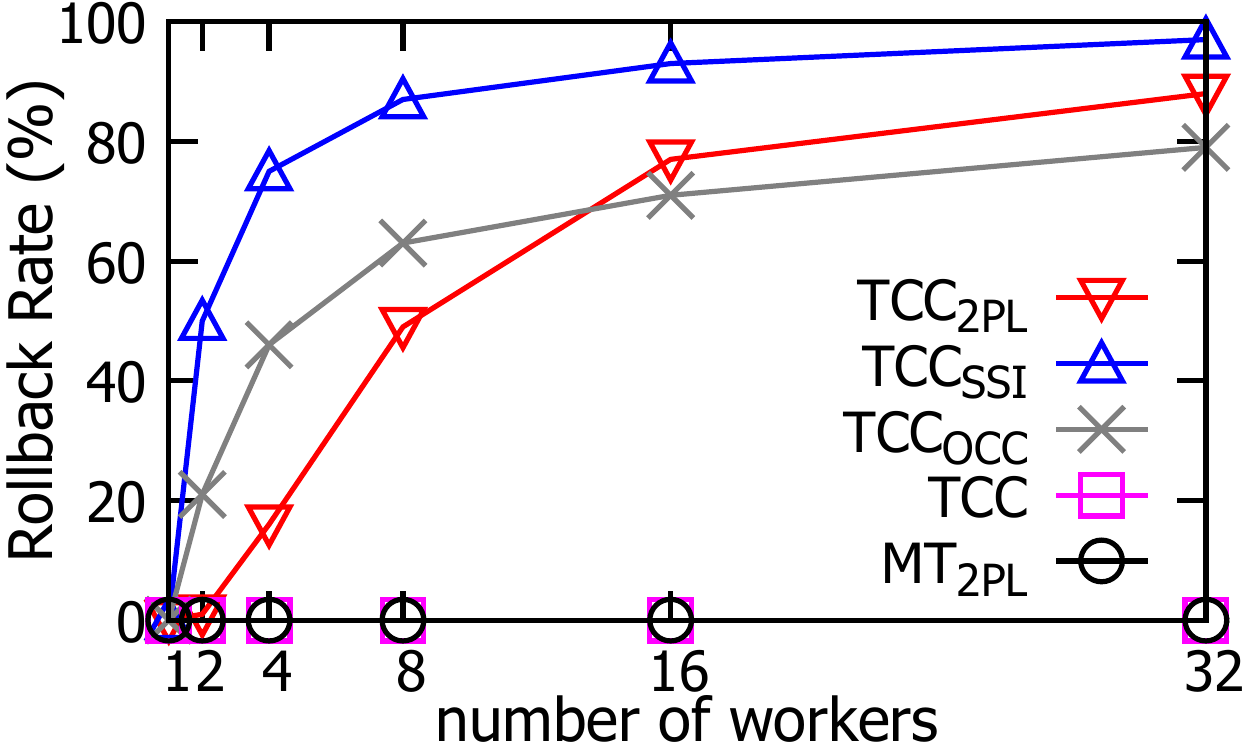}
                \label{fig:sequential_insert_abort}
            }
            \vspace{-0.5cm}
            \caption*{\small II. Sequential Insertion.}
            \label{fig:btree:seq}
            \vspace{-0.2cm}
      \end{figure}
  \end{minipage}
  \vspace{-0.2cm}
  \caption{\small Performance on B-tree Insertion.}
  \vspace{-0.2cm}
  \label{fig:btree_index}
\end{figure*}

Figure~\ref{fig:btree_index} shows the results on B-tree insertion. We can see that all the CC mechanisms perform similarly well when the degree of contention is low. When the degree of contention increases, the performance of $TCC_{2PL}$, $TCC_{SSI}$ and $TCC_{OCC}$ gradually becomes unbearable.
In the case of high contention, $TCC_{2PL}$ suffers from deadlocks. A large number of deadlocks was incurred when it upgraded the latches on the last leaf node of the B-tree (from shared mode to exclusive mode).
This leads to high deadlock-resolving cost and high abort rate (Figure~\ref{fig:btree_index}(d)).
While $TCC_{SSI}$ and $TCC_{OCC}$ do not need to deal with deadlocks, they suffer from high abort rates. When transactions are contending for the last leaf node, $TCC_{OCC}$'s validation phases will be highly likely to fail, and $TCC_{SSI}$ will encounter a large number of write-write conflicts, which can easily force transactions to abort.

In contrast, $TCC$'s performance is significantly better in the high-contention case. It performed as well as Shore-MT's built-in B-tree scheduler. The operational scheduler of $TCC$ is progressive. When a B-tree insertion fails, it automatically retries it, without aborting the host transaction. More importantly, it can learn from errors, such that the retries are limited.
As Table~\ref{table:retry} shows, even when the degree of contention is maximized, $TCC$ can complete a B-tree insertion with $1.7$ retries on average.

\begin{table}
\centering
\caption{\small Retry Frequency per Operation.}
\vspace{-0.2cm}
\label{table:retry}
\small{
\begin{tabular}{c||cccccc}
\multirow{2}{*}{} & \multicolumn{6}{|c}{\# of Workers} \\ 
& 1 & 2 & 4 & 8 & 16 & 32 \\ \hline\hline
B-tree Insert & 0 & 0.04 & 0.93 & 1.27 & 1.55 & 1.70 \\
Corner Case & 0 & 1.32 & 1.39 & 1.41 & 1.43 & 1.47 \\
\end{tabular}
}
\vspace{-0cm}
\end{table}


In our experiments on the corner case, we created an artificial operation in Shore-MT. There are two execution routes. When invoked, the operation will randomly choose one of the routes to execute.
In the first route, the operation is supposed to first read the block $A$, and then perform a large number of random reads, and finally update the block $B$.
In the second route, the operation is supposed to first read $B$, and then perform a large number of random reads, and finally update $A$.
The corner case is intentionally designed to handicap the generic CC mechanisms, including 2PL, SSI and OCC.


Figure~\ref{fig:corner} shows the results.
As we can see, when the degree of concurrency reaches a certain level, $TCC_{2PL}$, $TCC_{SSI}$ and $TCC_{OCC}$ all seem to be subject to starvation.
To $TCC_{2PL}$, a transaction can easily be involved in deadlocks. To $TCC_{SSI}$, write-write conflicts and anti-dependencies will be common, making transactions difficult to succeed.
To $TCC_{OCC}$, validation is difficult to pass.
In contrast, $TCC$ performs much better, as its operational scheduler is progressive. If an operation fails on a data block in the previous execution, it will latch the block upfront to avoid failing again. After one to two retries, the operation is guaranteed to succeed. According to Table~\ref{table:retry}, $TCC$ needs less than 1.47 retries to complete an operation.

\begin{figure}
\vspace{-0cm}
\hspace{-0.3cm}
\subfigure[{\small Throughput}]
{
    \includegraphics[width=0.5\linewidth]{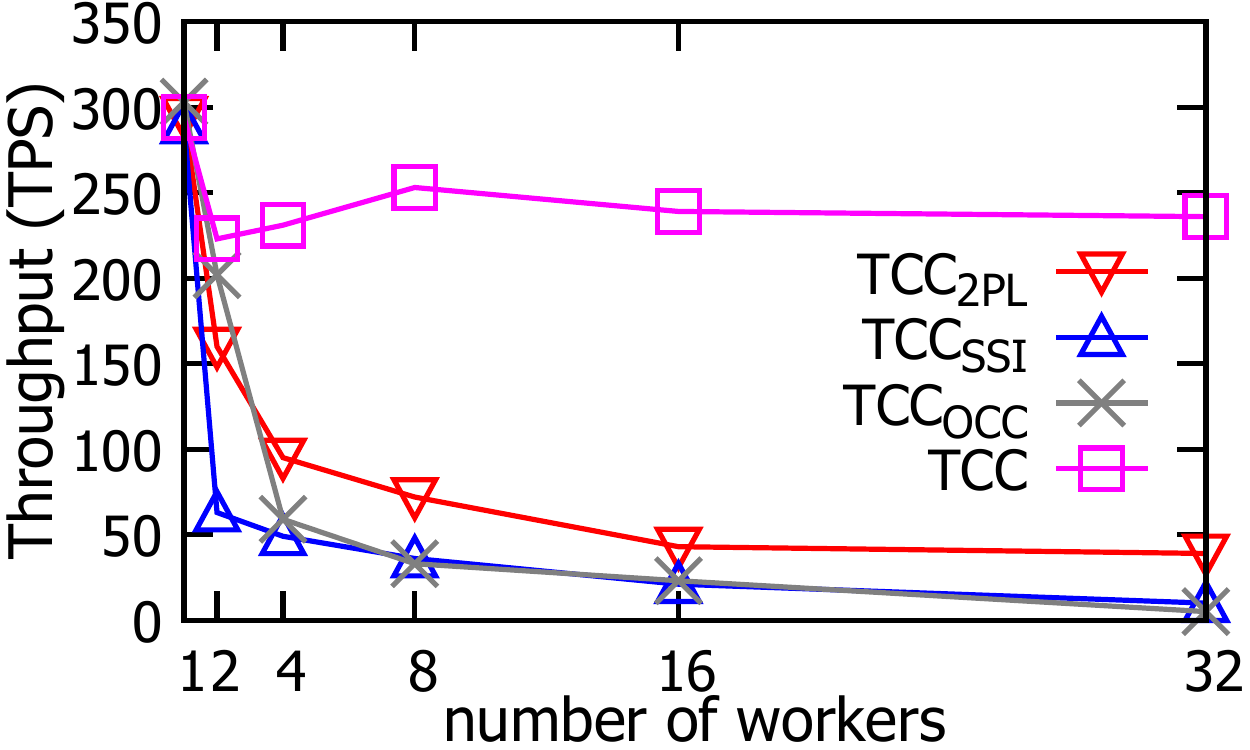}
}
\hspace{-0.4cm}
\subfigure[{\small Abort Rate}]
{
    \includegraphics[width=0.5\linewidth]{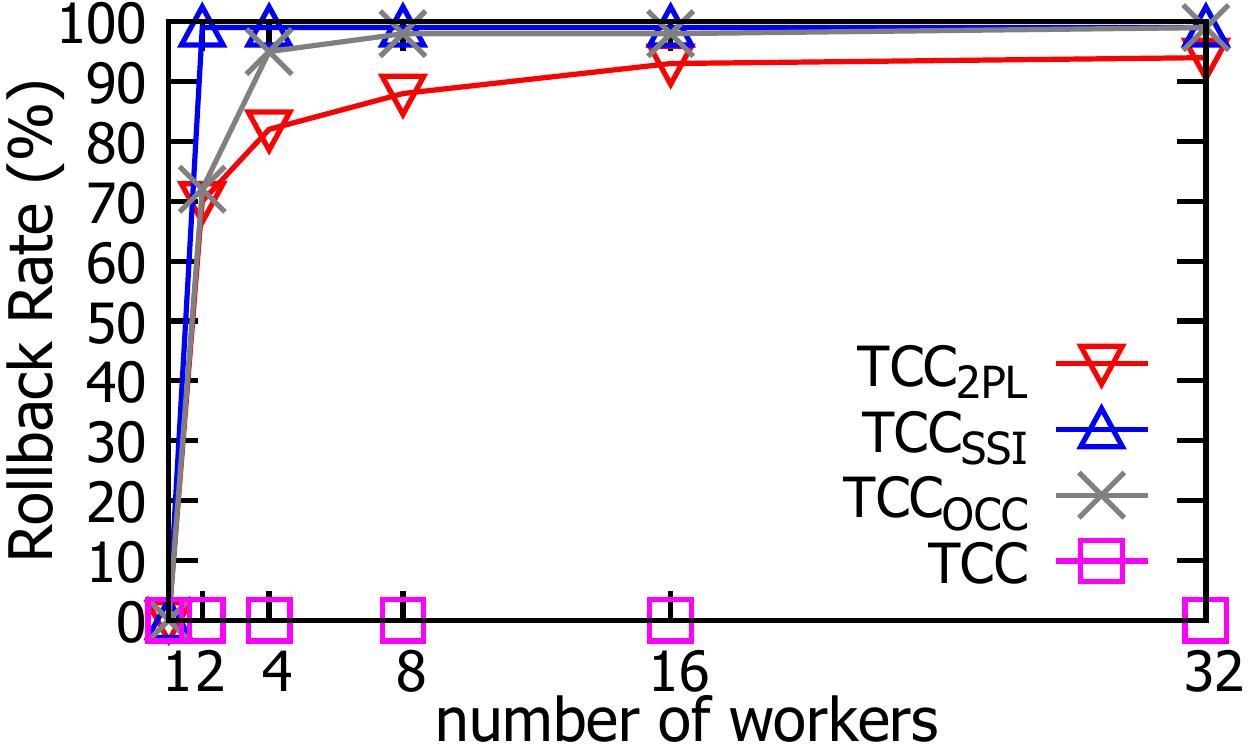}
}
\vspace{-0.4cm}
    \caption{\small Performance on a Corner Case.}
    \label{fig:corner}
\vspace{-0.4cm}
\end{figure}

The experiments justified our initiative to create a progressive operational scheduler. Generic CC mechanisms such as 2PL, SSI and OCC can perform fairly well in some cases. However, there are always cases they stop performing. It is unlikely to get rid of all such corner cases, if we are blind to data access patterns. This is known as the predictability gap.
In contrast, a progressive scheduler seems way more robust. It learns by doing, and is able to exploit the learned access patterns to improve the efficiency.
This is especially meaningful to TCC, which is supposed to make CC transparent to the rest of the system.

\vspace{-0.1cm}
\subsection{Experiments on Transactional Scheduler}
\vspace{-0.0cm}
\label{exp:tx_scheduler}

Our second set of experiments was conducted on the transactional scheduler. It mainly aimed to understand whether data semantics (i.e., commutative and inverse operations) can be exploited to improve performance. We used two types of workload, a revised New-Order workload of TPC-C and an artificial workload. We made two modifications on the New-Order transactions. First, we rebuilt the index of the \emph{order-line} table. The new index key is composed of four fields -- \emph{OrderId}, \emph{WarehouseId}, \emph{DistrictId} and \emph{OrderNumber}.
With this arrangement, insertions in the \emph{order-line} table will contend for the same B-tree leaf node.
Second, we made sure that there were $4$ insertions to the \emph{order-line} table in each transaction. This modification can enlarge the performance gaps among the $TCC$ variants.

We made the following data semantics explicit to $TCC$. First, tuple insertions are mutually commutative. Second, given the same tuple id, tuple deletion is the inverse operation of tuple insertion.

Figure~\ref{fig:new_order} shows the experiment results of revised New-Order. We can see that $TCC$ and $MT_{2PL}$ beat the other approaches.
$TCC_{2PL}$, $TCC_{SSI}$ and $TCC_{OCC}$ suffer from high abort rates, due to the same reason as that in the sequential B-tree insertion experiments.
While $TCC_{basic}$ does not consider data semantics, it still outperforms $TCC_{2PL}$, $TCC_{SSI}$ and $TCC_{OCC}$, due to the adoption of the progressive operational scheduler.
However, it is inferior to $TCC$. If an uncommitted transaction has inserted into a leaf node of a B-tree,  $TCC_{basic}$ will abort other transactions attempting to insert into the same leaf node, as they are accessing uncommitted data. Otherwise, it cannot ensure the recoverability of transactions. $TCC$ can avoid such abortion. $TCC$ allows its B-tree insertions to access uncommitted data, while still preserving recoverability. This is because B-tree insertions are reversible by invoking their inverse operations, i.e., B-tree deletion.
We can also see that $MT_{SSI}$ and $MT_{OCC}$ cannot achieve the same performance as $MT_{2PL}$. Both $MT_{SSI}$ and $MT_{OCC}$ suffer from high abort rates, which are incurred by conflicts on the \emph{district} table.

\begin{figure}
\vspace{-0.1cm}
\hspace{-0.3cm}
\subfigure[{\small Throughput}]
{
    \includegraphics[width=0.5\linewidth]{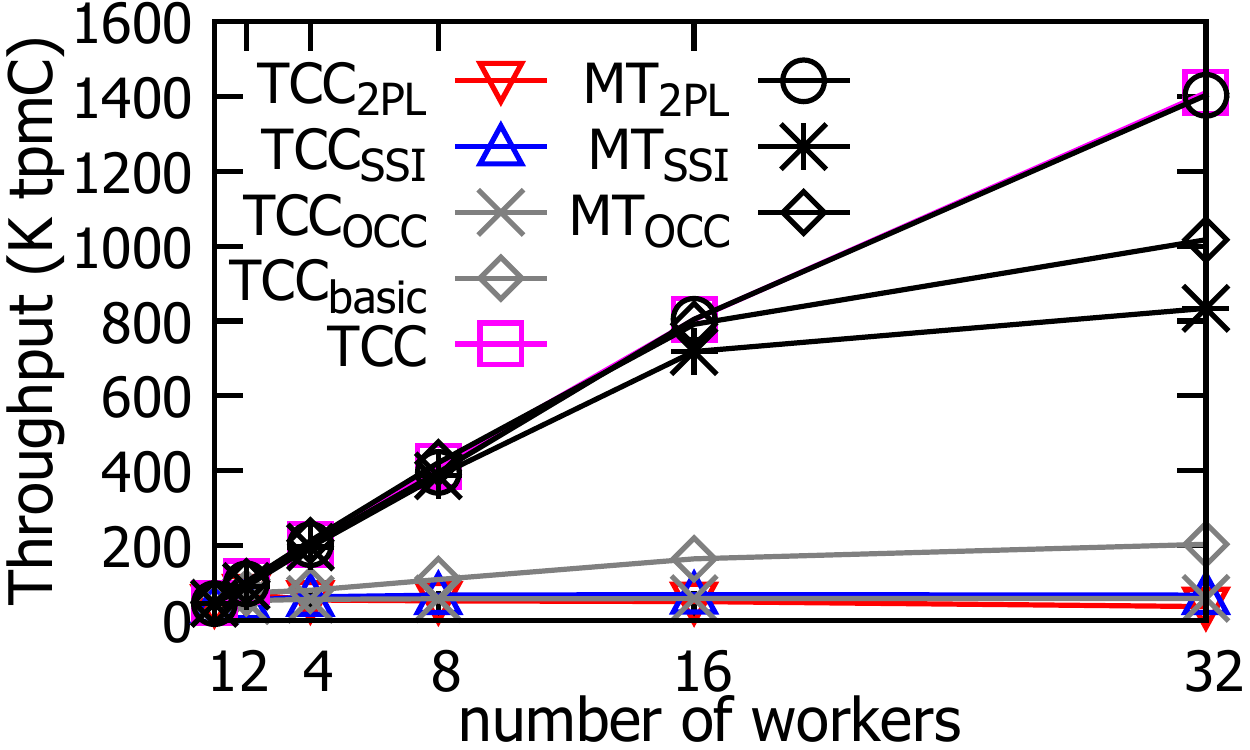}
}
\hspace{-0.4cm}
\subfigure[{\small Abort Rate}]
{
    \includegraphics[width=0.5\linewidth]{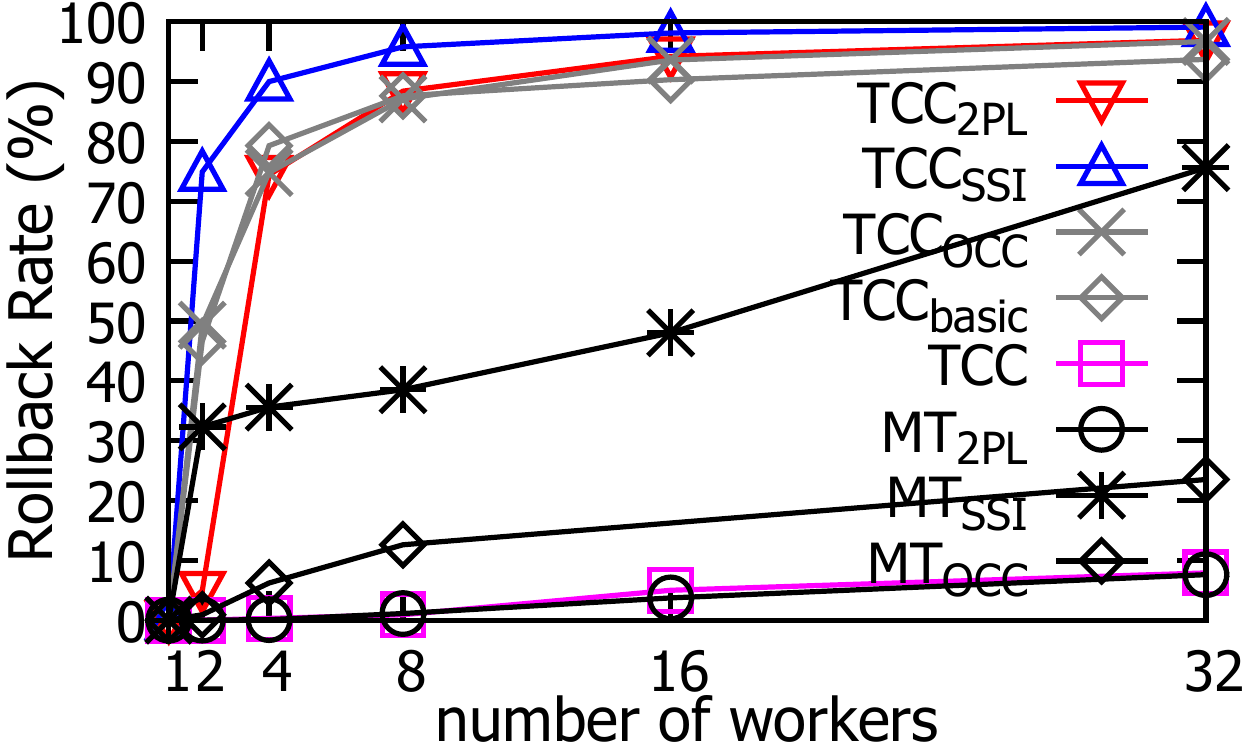}
	\label{fig:new_order:abort}
}
\vspace{-0.4cm}
    \caption{\small Performance on Revised New-Order Transactions.}
    \label{fig:new_order}
\vspace{-0.4cm}
\end{figure}

Our artificial workload was designed to demonstrate the difference between $TCC$ and $TCC_{basic}$. It contains two types of transactions. A short transaction is composed of $2$ B-tree insertions. A long transaction is composed of $8$ B-tree insertions. All insertions attempt to insert into the last leaf node of a B-tree. We ran the two types of transactions separately.
Figure~\ref{fig:semantic} shows the results. As we can see, $TCC$ achieved comparable performance as the original Shore-MT on both types of transactions.
$TCC_{basic}$ performed significantly worse than $TCC$, especially in the case of long transactions.
As $TCC$ considers commutativity of B-tree insertions, it allows multiple transactions to insert into the same B-tree leaf node concurrently.
In contrast, $TCC_{basic}$ does not allow such concurrency. When one transaction is performing the insertion, the other concurrent transactions have to be aborted. The longer the transactions, the higher the abort rate.

Therefore, we can conclude that data semantics can be powerful for enhancing the performance of TCC. Especially for data operations that are prone to confliction, it seems crucial to make them commutative and reversible (through inverse operations).

\begin{figure}[t]
      \centering
      \vspace{-0.5cm}
      \begin{minipage}{\linewidth}
          \begin{figure}[H]
              \subfigure[\small Throughput]
                {
                    \includegraphics[width = 0.49\linewidth]{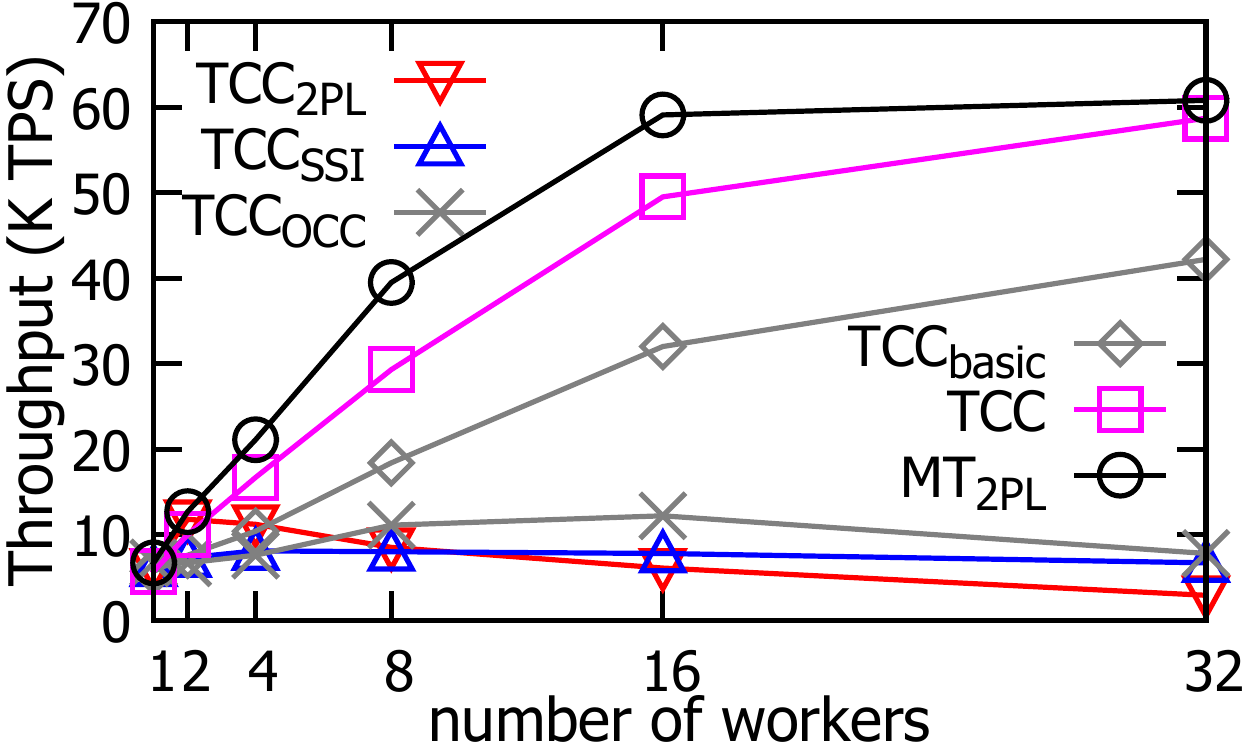}
                }
                \hspace{-0.5cm}
                \subfigure[\small Abort Rate]
                {
                    \includegraphics[width=0.49\linewidth]{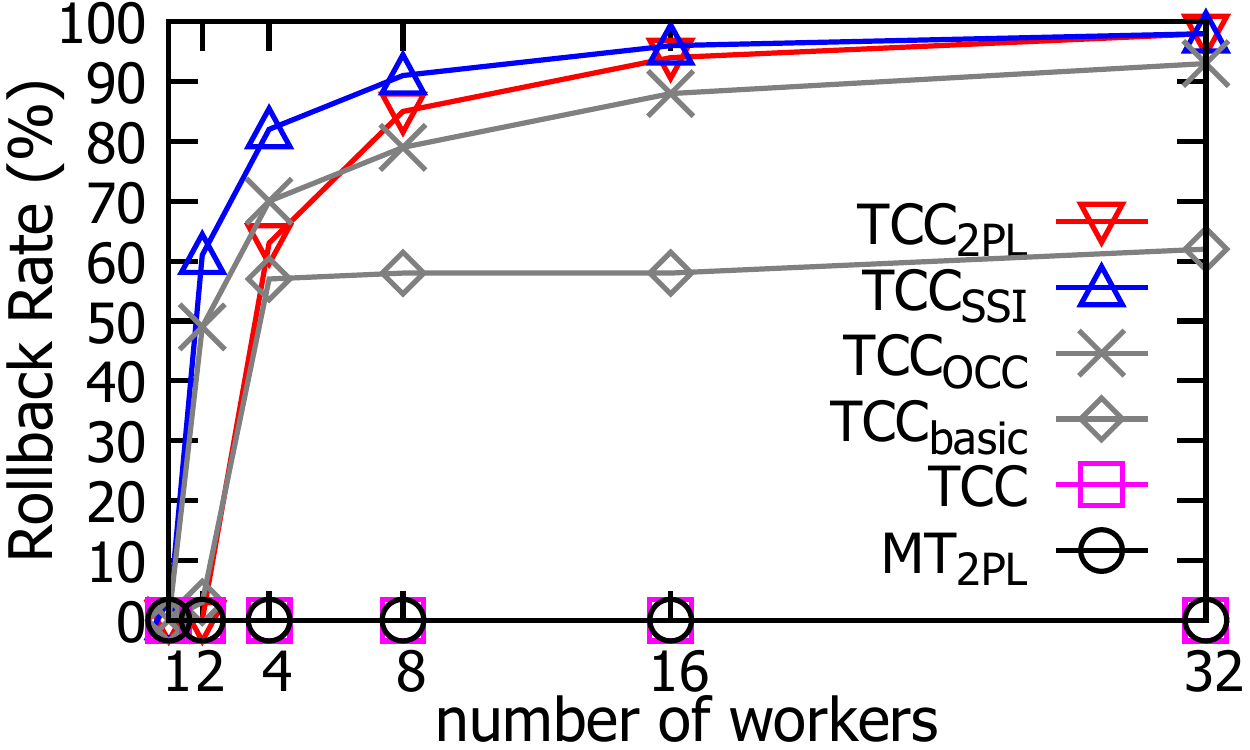}
                }
                \vspace{-0.5cm}
                \caption*{\small I. Short Transaction case.}
                \vspace{-0.6cm}
                \label{fig:tx1}
          \end{figure}
      \end{minipage}
      \hspace{-0.5cm}
      \begin{minipage}{\linewidth}
          \begin{figure}[H]
             \subfigure[\small Throughput]
                {
                    \includegraphics[width = 0.49\linewidth]{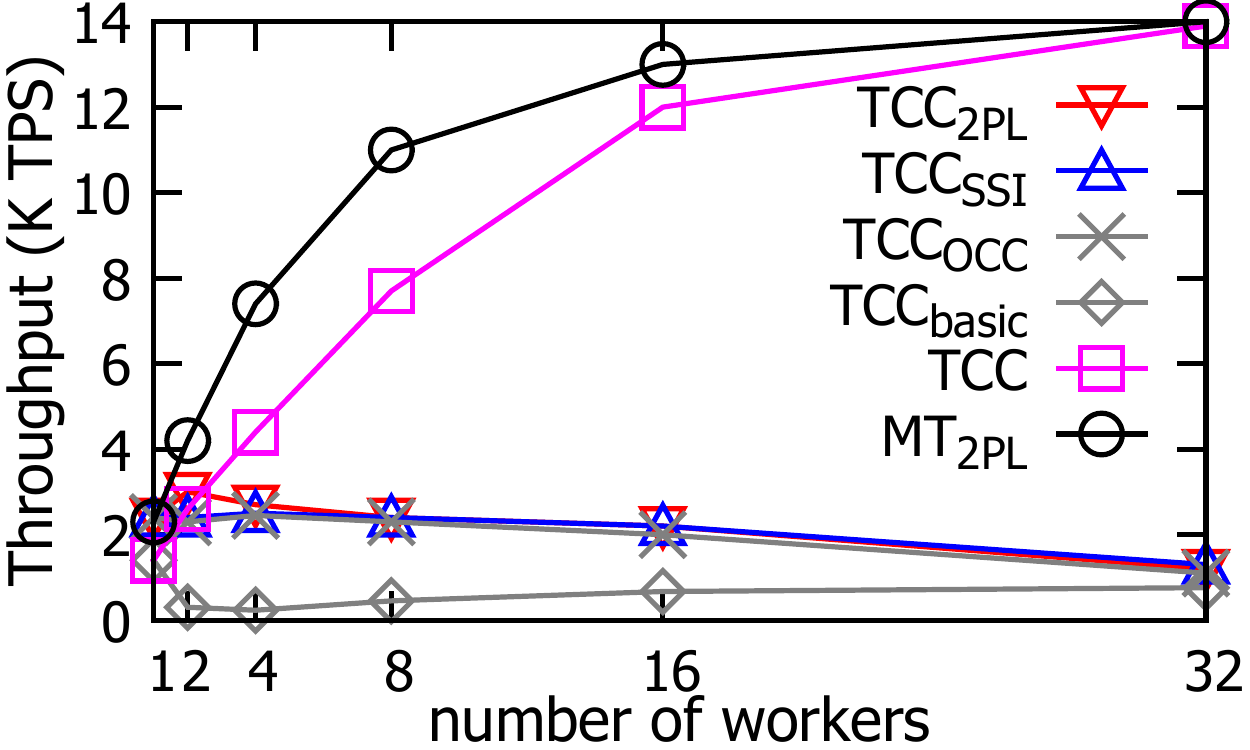}
                }
                \hspace{-0.5cm}
                \subfigure[\small Abort Rate]
                {
                    \includegraphics[width=0.49\linewidth]{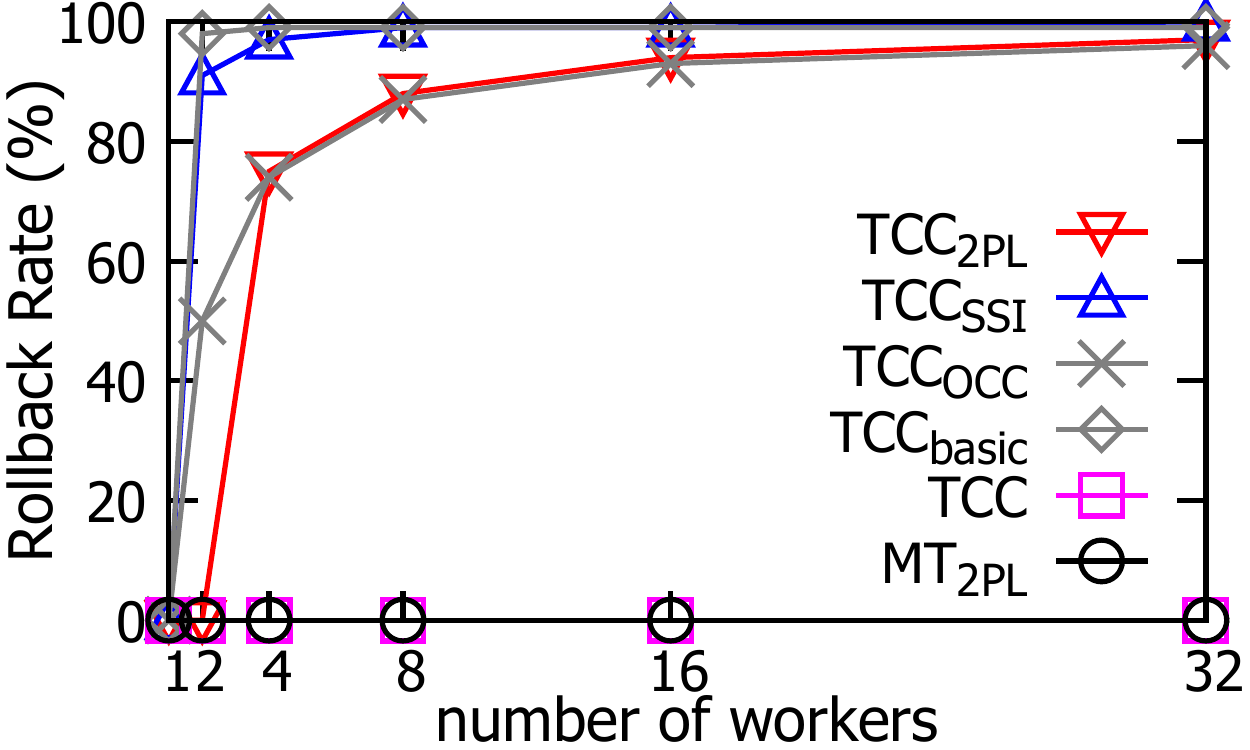}
                }
                \vspace{-0.5cm}
                \caption*{\small II. Long Transaction case.}
                \vspace{-0.2cm}
                \label{fig:tx16}
          \end{figure}
      \end{minipage}
      \vspace{-0.2cm}
      \caption{\small Performance on Short/Long Transactions.}
      \vspace{-0.3cm}
      \label{fig:semantic}
\end{figure}

\begin{figure}
\vspace{-0cm}
\hspace{-0.3cm}
\subfigure[{\small Throughput}]
{
    \includegraphics[width=0.5\linewidth]{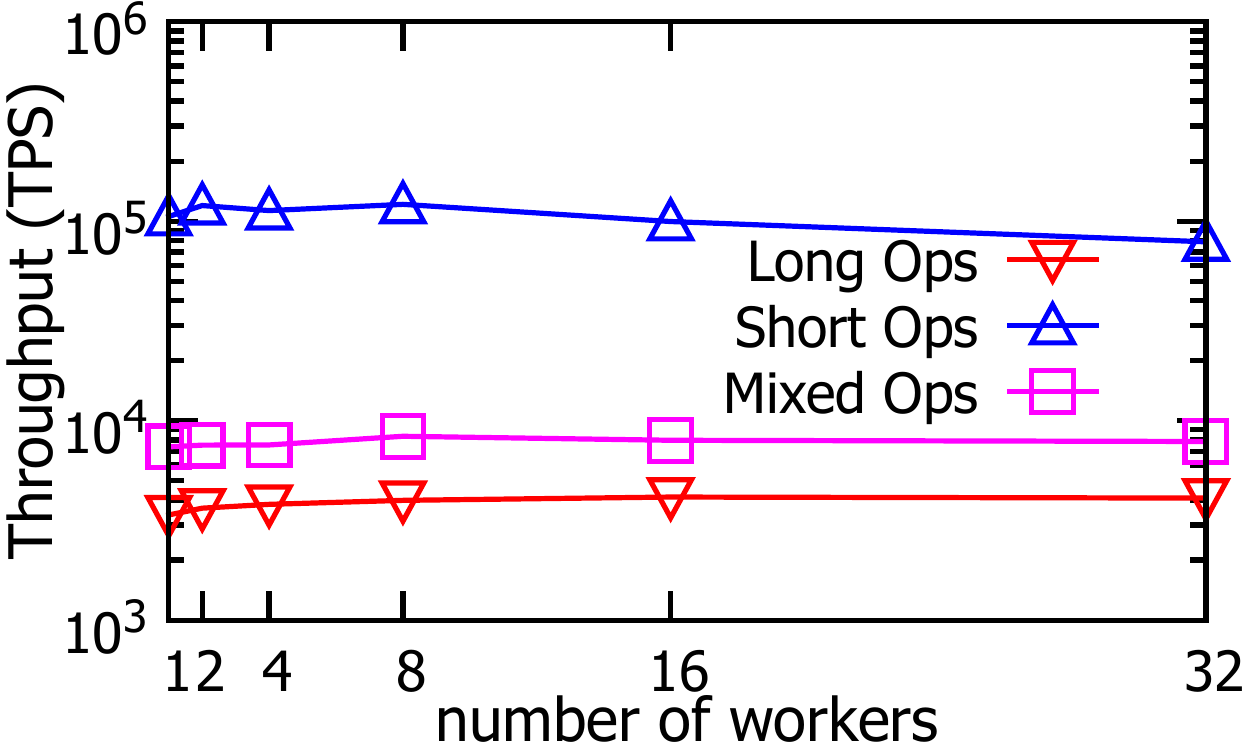}
}
\hspace{-0.4cm}
\subfigure[{\small Abort Rate}]
{
    \includegraphics[width=0.5\linewidth]{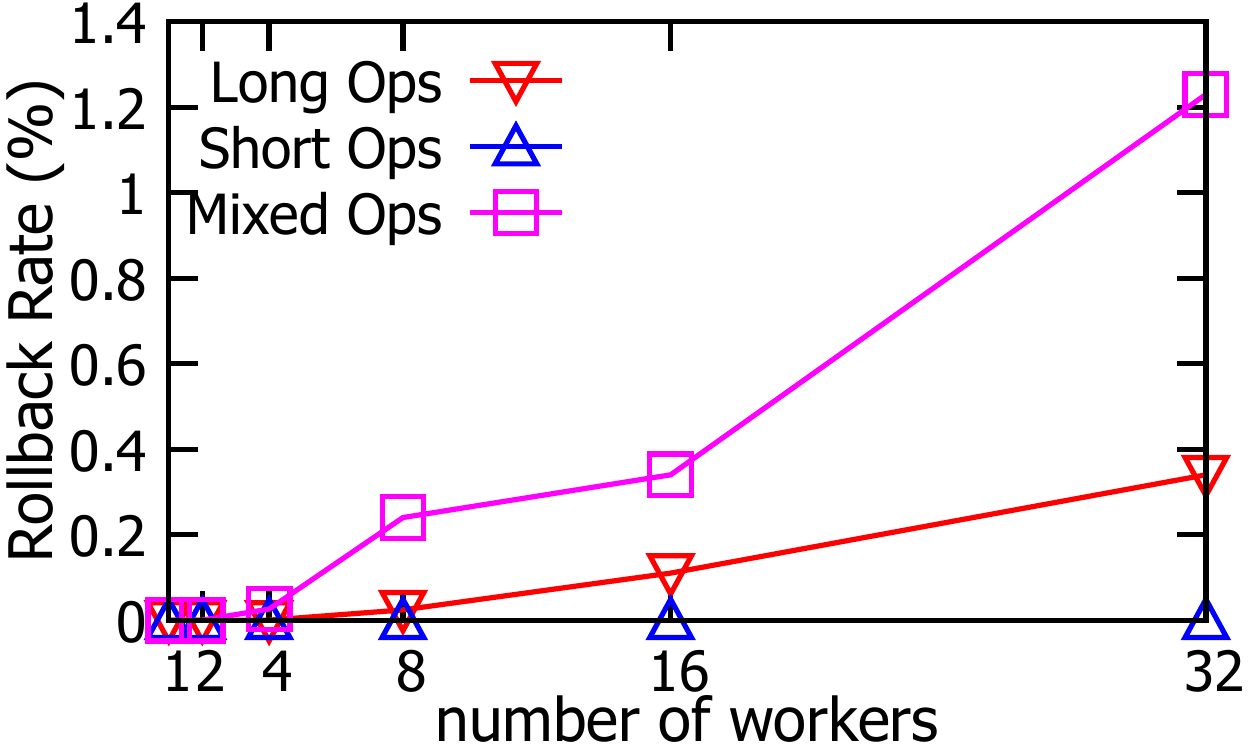}
}
\vspace{-0.5cm}
    \caption{\small Chance of Extra Aborts.}
    \label{fig:false_abort}
\vspace{-0.4cm}
\end{figure}

As TCC performs locking only after the operational latches are released, it may lead to extra aborts (Algorithm~\ref{alg:op_exec} Line~\ref{alg:op:10}). We used three types of workload to measure the abort rate caused by the separation between the latching and locking phases. We used short and long data operations. Short operations update a single record. Long operations update a set of $100$ records. In the first type of workload, each transaction consists of a short operation. In the second type of workload, each transaction consists of a long operation. In the third type of workload, each transaction consists of a randomly selected short or long operation.

Figure~\ref{fig:false_abort} shows the results of the experiments on the three types of workload. We can see that the mixed workload is more likely to incur abort. A long operation provides a relatively large window between the latching phase and the locking phase. This gives short operations more chance to jump the order and incur abort. Nevertheless, such abort is not a serious concern to TCC. As shown in Figure~\ref{fig:false_abort}, it does not occur frequently even in the worst case.


\vspace{-0.1cm}
\subsection{Experiments on OLTP Benchmarks}
\vspace{-0.0cm}

Our final set of experiments was conducted on the benchmarks of TATP and TPC-C.

For the experiments on TATP, we set the scale factor to $10$. In each test, we ran the TATP workload for more than 10 minutes. We increased the number of worker threads to see how the system scales. Figure~\ref{fig:tatp} shows the results of the experiments.
As the degree of contention is low in TATP, all CC mechanisms scale quite well. We could not see significant difference among the different approaches.

\begin{figure}
\hspace{-0.3cm}
\subfigure[{\small Throughput}]
{
    \includegraphics[width=0.5\linewidth]{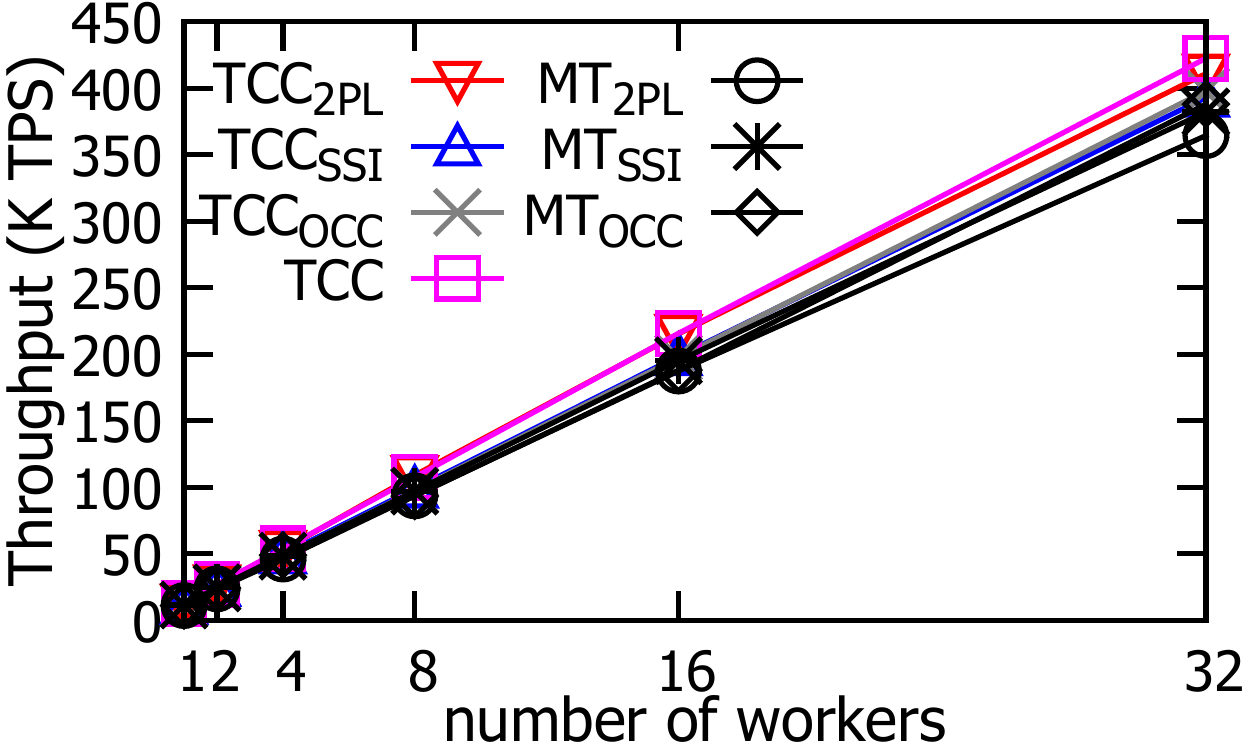}
}
\hspace{-0.4cm}
\subfigure[{\small Abort Rate}]
{
    \includegraphics[width=0.5\linewidth]{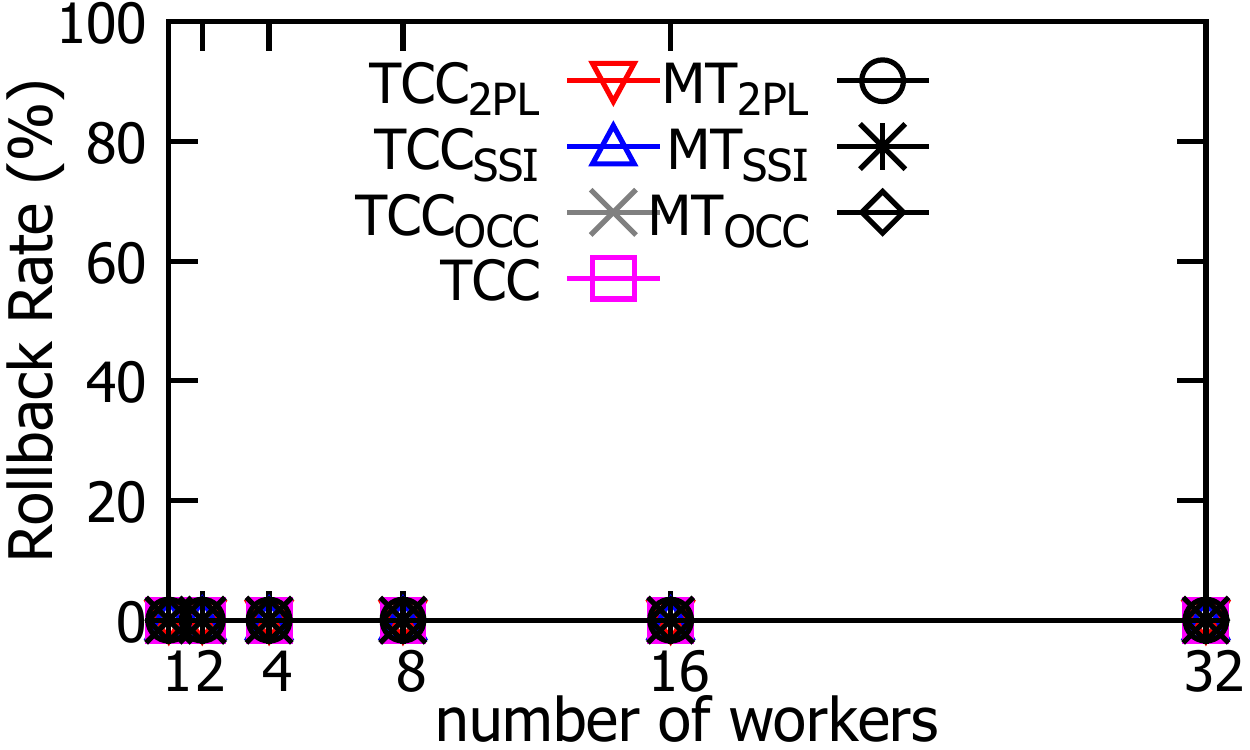}
    \label{fig:tatp_abort_32}
}
\vspace{-0.5cm}
    \caption{\small Performance on TATP.}
    \label{fig:tatp}
\vspace{-0.3cm}
\end{figure}

For the experiments on TPC-C, we set the scale factor to $10$. In each test, we ran the standard TPC-C workload (without wait time) for 10 minutes. We also increased the number of worker threads to evaluate the scalability. Figure~\ref{fig:tpcc} shows the results of the experiments.



We can see that most of the CC mechanisms achieved relatively good performance on TPC-C, except $TCC_{2PL}$. $TCC_{2PL}$ scales well when there are less than $8$ workers. When the degree of concurrency exceeds 8, its throughput drops quickly. This is mainly due to that $TCC_{2PL}$ cannot deal with ``select-for-update'' request. $TCC_{2PL}$ has no concept of operation. When encountering ``select-for-update'', it cannot predict that the data blocks accessed by ``select'' will be subsequently ``updated''. Thus, it had to frequently perform lock upgrades, which led to a large number of deadlocks.
In contrast, $TCC$ is able to deal with the ``select-for-update'' semantics. When encountering ``select-for-update'', the data organization tier can explicitly tell $TCC$ that the corresponding operation should place exclusive locks on the data blocks it has accessed. Then, $TCC$ can avoid lock upgrade. As $TCC_{SSI}$ and $TCC_{OCC}$ do not perform locking, they do not suffer from the lock upgrading problems.

Comparing the three Shore-MT mechanisms, we can find that $MT_{2PL}$ performs slightly worse than $MT_{SSI}$ and $MT_{OCC}$. $MT_{2PL}$ mainly suffers from the implementation of its predicate locks. When a transaction accesses indexed records in the \emph{warehouse} and \emph{district} tables, it will place predicate locks. The predicate locks are first shared locks. When updates are performed, they are upgraded to exclusive locks. Lock upgrade can cause deadlocks, which affect $MT_{2PL}$'s performance.

In fact, we found that 2PL in general do not perform as well as SSI in TPC-C. The Payment transactions of TPC-C always need to update the \emph{warehouse} table, while the New-Order transactions always need to read the \emph{warehouse} table. When 2PL is adopted, a large number of transactions will be blocked by the read-write conflicts. In contrast, the SSI approaches do not face this problem. $TCC$ adopts 2PL as its transactional scheduler. In TPC-C, it cannot perform as well as $TCC_{SSI}$. Nevertheless, $TCC$ is superior to $TCC_{SSI}$ in robustness. As shown in our previous experiments, $TCC_{SSI}$ can exhibit very poor performance in a variety of cases. From this perspective, no one of $TCC_{SSI}$, $TCC_{2PL}$ and $TCC_{OCC}$ can compare to $TCC$.


\begin{figure}[t]
      \centering
             \subfigure[Throughput]
                {
                    \includegraphics[width = 0.49\linewidth]{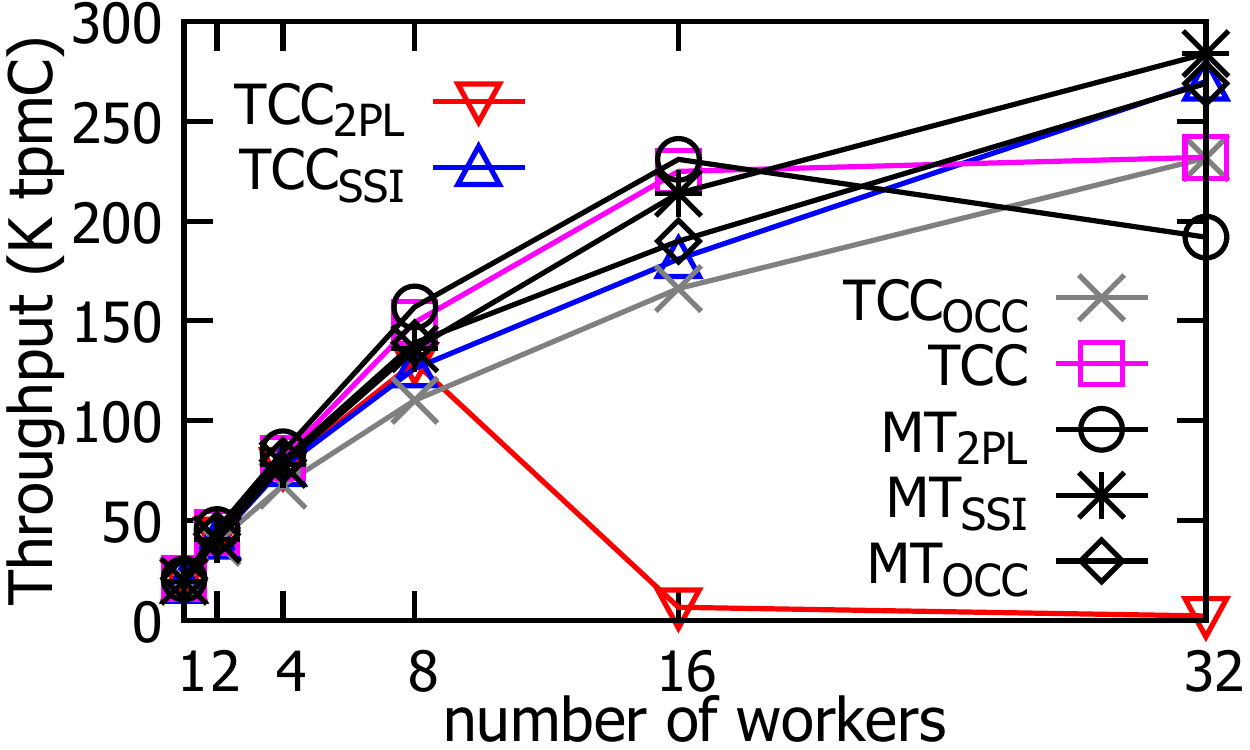}
                    \label{fig:revised_tpcc_throughput}
                }
                \hspace{-0.5cm}
                \subfigure[Abort Rate]
                {
                    \includegraphics[width=0.49\linewidth]{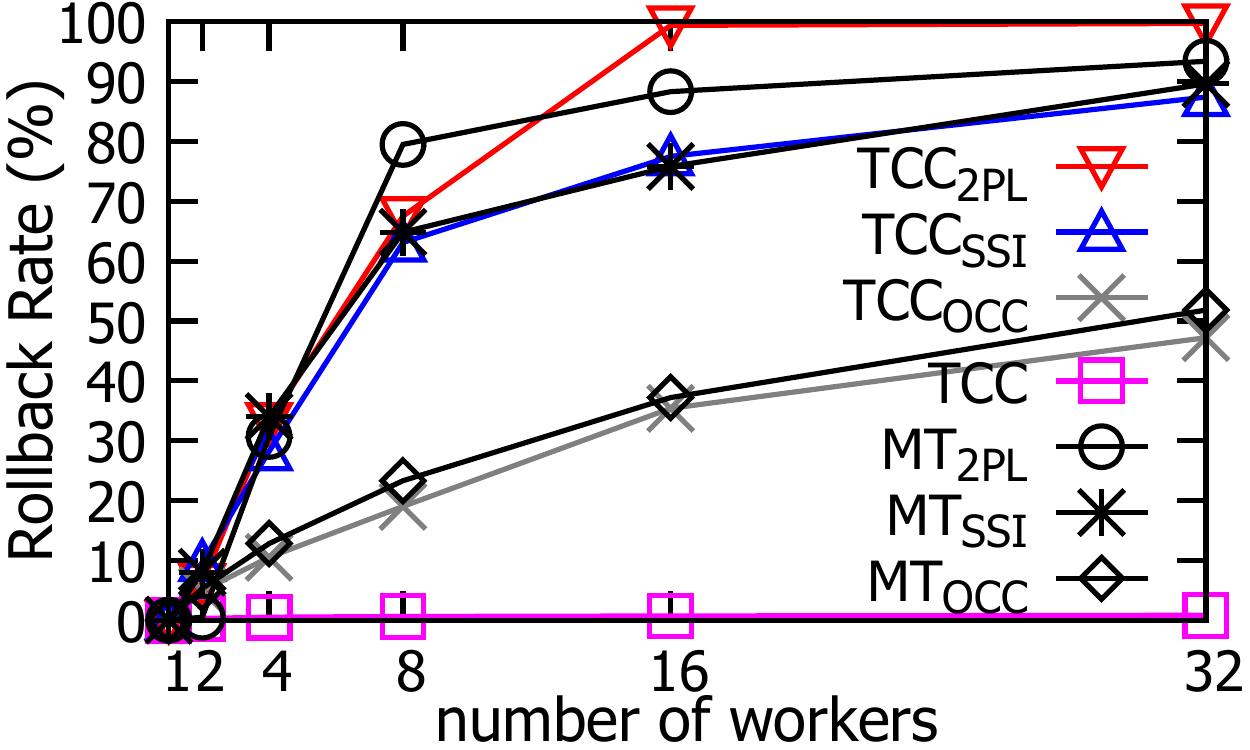}
                    \label{fig:revised_tpcc_abort}
                }
      \vspace{-0.5cm}
      \caption{\small Performance on TPCC.}
      \vspace{-0.3cm}
      \label{fig:tpcc}
\end{figure}

The experiments show that when TCC is taking care of the concurrency control of index structures, a DBMS can process transactions efficiently.
The good performance of TCC is attributable to both its robust operational scheduler and its ability to utilize data semantics.

\vspace{-0.1cm}
\section{Conclusion}
\vspace{-0cm}

In this paper, we attempted to separate the layer of concurrency control from a DBMS. Our results showed that the separation is feasible, at least on the indexes of a DBMS. On the one hand, transactional safety can be guaranteed. On the other hand, the performance issues caused by the separation is controllable. We believe that the separation will be enormously beneficial, as it can substantially improve the flexibility of a DBMS. With such flexibility, a DBMS will be easier to implement, modify and extend.

To make the separation work, it is important to have a progressive scheduler that is robust against unpredictable data accesses. It is also important to allow the DBMS to declare data semantics to the CC layer, especially on data operations that are prone to confliction. To achieve these, we created TCC, which can deal the the predictability and semantic gaps effectively.


However, further research is required to make TCC practical. First, TCC needs to tested in a broader scope of scenarios. In this paper, we evaluated it on the indexes of a real-world DBMS. Its applicability on an entire DBMS, especially its components on metadata management and space management, requires further investigation.
Second, a transparent recovery mechanism should be integrated with TCC to support full-scale ACID. Third, some principles need to be identified to help system developers make good use of TCC, including the guidelines on how to determine the granularity of data operations, how to create commutative and inverse operations, etc.

\vspace{-0.1cm}
\bibliographystyle{abbrv}
\bibliography{vldb_sample}  

\balance



\end{document}